\documentclass[runningheads, envcountsame]{llncs}

\usepackage{booktabs} 
\usepackage[ruled]{algorithm2e} 
\usepackage{pgfplots}
\pgfplotsset{compat=1.16}

\SetAlFnt{\small}
\SetAlCapFnt{\small}
\SetAlCapNameFnt{\small}
\SetAlCapHSkip{0pt}
\IncMargin{-\parindent}

\usepackage{amsmath, amsfonts, amssymb}
\usepackage{hyperref}
\usepackage[appendix=append]{apxproof}

\DeclareMathAlphabet{\mathcal}{OMS}{cmsy}{m}{n}

\usepackage[normalem]{ulem} 

\newcommand{\E}{\mathbb{E}}

\renewcommand{\P}{\mathbb{P}}

\newcommand{\pe}[2][]{\ensuremath{\E_{#1}[#2]}}
\newcommand{\pel}[2][]{\ensuremath{\E_{#1}\left[#2\right]}}

\usepackage{bm}
\newcommand{\mmod}[1]{\bm{#1}}
\newcommand{\vals}{\mmod{v}}
\newcommand{\bids}{\mmod{b}}
\newcommand{\pays}{\mmod{p}}
\newcommand{\fee}{\ensuremath{f}}
\newcommand{\Bids}{\mmod{\sigma}}
\newcommand{\opt}{\mmod{x}^*}
\newcommand{\alloc}{\mmod{x}}
\newcommand{\pne}{\ensuremath{\mathrm{PNE}}}
\newcommand{\mne}{\ensuremath{\mathrm{MNE}}}
\newcommand{\mce}{\ensuremath{\mathrm{CE}}}
\newcommand{\mcce}{\ensuremath{\mathrm{CCE}}}

\newcommand{\SW}{\text{SW}}

\mathchardef\mhyphen="2D
\newcommand{\yDPA}{\ensuremath{\gamma\mhyphen\textsc{HYA}}}

\newcommand{\yFPA}{\ensuremath{\gamma\mhyphen\textsc{FPA}}}

\newcommand{\fpa}{\ensuremath{\textsc{FP-Auction}}}
\newcommand{\spa}{\ensuremath{\textsc{SP-Auction}}}
\newcommand{\ycfpa}{\ensuremath{\gamma\mhyphen\text{corrupt  auction}}}

\newcommand{\mech}{\ensuremath{\mathcal{M}}}

\newcommand{\nob}{\ensuremath{\text{NOB}}}
\newcommand{\yhybrid}{$\gamma$-hybrid auction}

\newcommand{\set}[1]{\ensuremath\{#1\}}

\newcommand{\poa}{\ensuremath{\mathrm{POA}}}

\newcommand{\ce}{\sigma}
\newcommand{\HB}{\textit{HB}}
\newcommand{\SB}{\textit{SB}}

\newcommand{\myparagraph}[1]{\smallskip\emph{#1}~}

\title{%
Corruption in Auctions: \\ Social Welfare Loss in Hybrid Multi-Unit Auctions
}
\titlerunning{Social Welfare Loss in Hybrid Multi-Unit Auctions}

\author{%
Andries~van~Beek\inst{1} 
\and 
Ruben~Brokkelkamp\inst{2}
\and 
Guido~Sch\"afer\inst{2, 3}
}

\institute{
CentER, 
Dept.~of Econometrics \& Operations Research, 
Tilburg University, The Netherlands \\
\email{a.j.vanbeek@tilburguniversity.edu}\\
\and
Networks and Optimization, Centrum Wiskunde \& Informatica (CWI), The Netherlands\\
\email{\{ruben.brokkelkamp,g.schaefer\}@cwi.nl} 
\and
Institute for Logic Language and Computation,
University of Amsterdam, The Netherlands
}

\newcommand{\BibTeX}{\rm B\kern-.05em{\sc i\kern-.025em b}\kern-.08em\TeX}

\begin{document}

\maketitle 

\begin{abstract}
We initiate the study of the social welfare loss caused by corrupt auctioneers, both in single-item and multi-unit auctions.
In our model, the auctioneer may collude with the winning bidders by letting them lower their bids in exchange for a (possibly bidder-dependent) fraction $\gamma$ of the surplus. We consider different corruption schemes. 
In the most basic one, all winning bidders lower their bid to the highest losing bid. We show that this setting is equivalent to a \emph{$\gamma$-hybrid auction} in which the payments are a convex combination of first-price and the second-price payments. More generally, we consider corruption schemes that can be related to \emph{$\gamma$-approximate first-price auctions (\yFPA)}, where the payments recover at least a $\gamma$-fraction of the first-price payments. Our goal is to obtain a precise understanding of the robust price of anarchy (POA) of such auctions. 
If no restrictions are imposed on the bids, we prove a bound on the robust POA of $\yFPA$ which is tight (over the entire range of $\gamma$) for the single-item and the multi-unit auction setting. 
On the other hand, if the bids satisfy the no-overbidding assumption a more fine-grained landscape of the price of anarchy emerges, depending on the auction setting and the equilibrium notion.
Albeit being more challenging, we derive (almost) tight bounds for both auction settings and several equilibrium notions, basically leaving open some (small) gaps for the coarse-correlated price of anarchy only.
\end{abstract}

\section{Introduction}

\myparagraph{Motivation and Background.}
We consider auction settings where a seller wants to sell some items and for this purpose recruits an auctioneer to organize an auction on their behalf.\footnote{Throughout this paper, we use ``they'' as the gender-neutral form for third-person singular pronouns.} 
Such settings are widely prevalent in practice as they emerge naturally whenever the seller lacks the expertise (or facilities, time, etc.) to host the auction themselves. 
For example, individual sellers usually involve dedicated auctioneers or auction houses when they want to sell particular objects (such as real estate, cars, artwork, etc.). 
In private companies, the responsible finance officers are typically in charge of handling the procurement auctions. Similarly, government procurement is usually executed by some entity that acts on behalf of the government. 
The dilemma in such settings is that the incentives of the seller and the auctioneer are rather diverse in general: while the seller is interested in extracting the highest payments for the objects (or getting service at the lowest cost), the agent primarily cares about maximizing their own gains from hosting the auction. Although undesirably, this misalignment leads (unavoidably) to fraudulent schemes which might be used by the auctioneer to manipulate the auction to their own benefit. 

Corruption in auctions, where an auctioneer engages in bid rigging with one (or several) of the bidders, occurs rather frequently in practice, especially in the public sector (e.g., in construction and procurement auctions). For example, in 1999 the procurement auction for the construction of the new Berlin Brandenburg airport had to be rerun after investigations revealed that the initial winner was able to change the bid after they had illegally acquired information about the application of one of their main competitors (see \cite{WSJ99}). As another example, in 1993 the New York City School Construction Authority caused a scandal when investigation revealed that they used a simple (but effective) bid-rigging scheme in a procurement auction setting (see \cite{NYT93}): 
\begin{quote}
\emph{``In what one investigator described as a nervy scheme, that worker would unseal envelopes at a public bid opening, saving for last the bid submitted by the contractor who had paid him off. At that point, knowing the previous bids, the authority worker would misstate the contractor's bid, insuring that it was low enough to secure the contract but as close as possible to the next highest bid so that the contractor would get the largest possible price.''}
\end{quote}

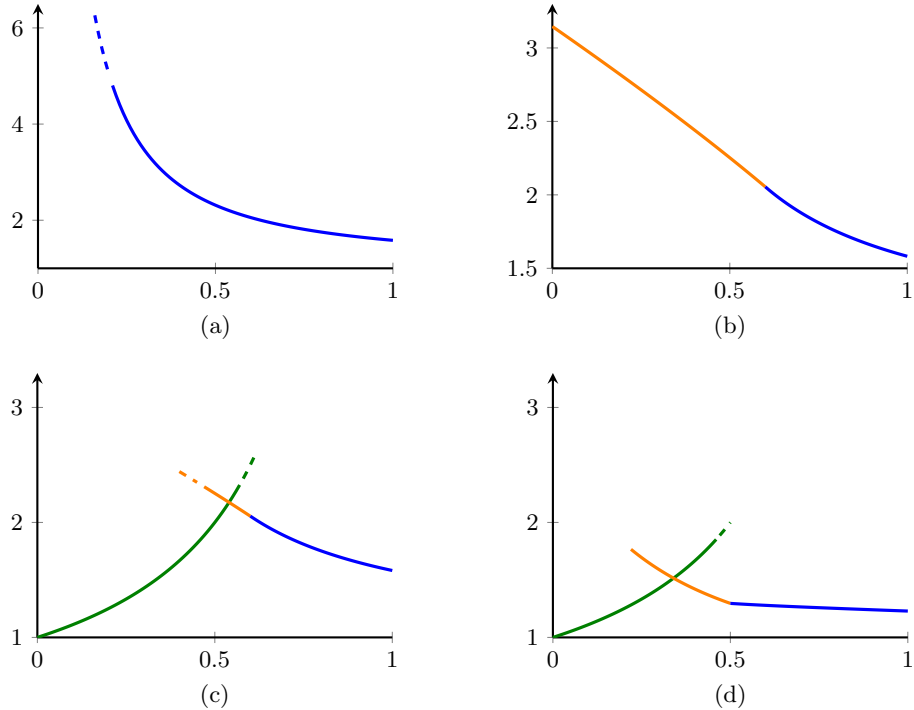
\begin{figure*}[t]%
\centering
\newcommand{\graphwidth}{6.3cm}%
\newcommand{\graphheight}{5.1cm}
\newcommand{\graphinbetween}{1cm}%
\begin{tikzpicture}
\begin{axis}[
        width=\graphwidth, 
        height=\graphheight,
        axis line style={thick},
        xlabel={(a)},
        ymin=1, ymax=6.5,
        xmin=0, xmax=1,
        axis x line={bottom},
        axis y line={left},
        ytick distance=2,
        xtick distance=0.5,
        x axis line style=-
    ]
\addplot[no marks, dashed, very thick, color=blue] coordinates {
(0.16, 6.262088674984048) 
(0.17, 5.898800076064398) (0.18000000000000002, 5.577116241553959) (0.19, 5.290557290824336) (0.2, 5.033918274531521) 
};
\addplot[no marks, very thick, color=blue] coordinates{ 
(0.21000000000000002, 4.802966811597772) 
(0.22, 4.594223823060196) (0.23, 4.404801763404024) (0.24000000000000002, 4.23228336839986) (0.25, 4.074629441455096) (0.26, 3.930107783555831) (0.27, 3.797237746956672) (0.28, 3.674746493503648) (0.29000000000000004, 3.561534134285204) (0.3, 3.4566456886334516) (0.31, 3.3592483369278714) (0.32, 3.2686128247203565) (0.33, 3.184098152765343) 
(0.34, 3.105138890384421) (0.35000000000000003, 3.0312345998254613) (0.36000000000000004, 2.961940971760537) (0.37, 2.896862357160519) (0.38, 2.8356454457837628) (0.39, 2.777973891620244) (0.4, 2.72356372458463) (0.41000000000000003, 2.672159418276287) (0.42000000000000004, 2.6235305077230646) (0.43, 2.5774686701840803) (0.44, 2.533785197416633) (0.45, 2.492308800153379) (0.46, 2.4528836955271824) (0.47000000000000003, 2.415367936313134) (0.48000000000000004, 2.3796319475089605) (0.49, 2.3455572412415315) (0.5, 2.3130352854993315) (0.51, 2.2819665059307543) (0.52, 2.252259403060211) (0.53, 2.2238297698733374) (0.54, 2.1965999969013024) (0.55, 2.170498453766444) (0.56, 2.1454589376973265) (0.5700000000000001, 2.1214201808295092) (0.5800000000000001, 2.098325409218629) (0.59, 2.0761219474374495) (0.6, 2.054760863434977) (0.61, 2.0341966490258425) (0.62, 2.0143869319699843) (0.63, 1.9952922161116033) (0.64, 1.9768756464849293) (0.65, 1.9591027966731223) (0.66, 1.94194147603454) (0.67, 1.9253615546949996) (0.68, 1.9093348044518619) (0.6900000000000001, 1.8938347539510443) (0.7000000000000001, 1.8788365566859022) (0.7100000000000001, 1.8643168705311133) (0.72, 1.8502537476684993) (0.73, 1.836626533887848) (0.74, 1.8234157763566685) (0.75, 1.8106031390503594) (0.76, 1.7981713251203177) (0.77, 1.7861040055534634) (0.78, 1.7743857535438372) (0.79, 1.7630019840564422) (0.8, 1.751938898116266) (0.81, 1.7411834314023216) (0.8200000000000001, 1.7307232067682339) (0.8300000000000001, 1.720546490348065) (0.8400000000000001, 1.7106421509391907) (0.85, 1.7009996223836499) (0.86, 1.69160886869583) (0.87, 1.6824603517080718) (0.88, 1.6735450010269943) (0.89, 1.6648541861124209) (0.9, 1.656379690307903) (0.91, 1.6481136866672679) (0.92, 1.6400487154354744) (0.93, 1.6321776630545954) (0.9400000000000001, 1.6244937425770176) (0.9500000000000001, 1.6169904753781783) (0.9600000000000001, 1.6096616740703684) (0.97, 1.6025014265275033) (0.98, 1.5955040809383239) (0.99, 1.5886642318123736) (1.0, 1.5819767068693265)
};
\end{axis}
\end{tikzpicture}
\hspace*{\graphinbetween}
%
\begin{tikzpicture}
\begin{axis}[
        width=\graphwidth, 
        height=\graphheight,
        axis line style={thick},
        xlabel={(b)},
        ymin=1.5, ymax=3.3,
        axis x line={bottom},
        axis y line={left},
        xtick distance=0.5,
        x axis line style=-
    ]
\addplot[no marks, very thick, color=orange] coordinates {
(0.0, 3.1461932206205825) (0.01, 3.1293747126511713) (0.02, 3.1125239216565728) (0.03, 3.095640380190383) (0.04, 3.078723609338155) (0.05, 3.061773118317541) (0.06, 3.044788404060239) (0.07, 3.027768950774747) (0.08, 3.010714229488804) (0.09, 2.9936236975703716) (0.1, 2.9764967982258947) (0.11, 2.959332959974506) (0.12, 2.9421315960967562) (0.13, 2.9248921040563243) (0.14, 2.90761386489307) (0.15, 2.8902962425856744) (0.16, 2.872938583381963) (0.17, 2.855540215094898) (0.18, 2.8381004463620285) (0.19, 2.820618565866073) (0.2, 2.8030938415140803) (0.21, 2.7855255195724364) (0.22, 2.7679128237547848) (0.23, 2.7502549542596553) (0.24, 2.7325510867543614) (0.25, 2.714800371301444) (0.26, 2.6970019312236198) (0.27, 2.6791548619028416) (0.28, 2.661258229508738) (0.29, 2.6433110696512414) (0.3, 2.625312385951799) (0.31, 2.6072611485270274) (0.32, 2.589156292378144) (0.33, 2.5709967156788784) (0.34, 2.552781277953888) (0.35000000000000003, 2.5345087981389653) (0.36, 2.51617805251346) (0.37, 2.497787772494435) (0.38, 2.479336642281022) (0.39, 2.4608232963362755) (0.4, 2.4422463166925557) (0.41000000000000003, 2.4236042300649854) (0.42, 2.4048955047559195) (0.43, 2.386118547331526) (0.44, 2.3672716990495153) (0.45, 2.3483532320147145) (0.46, 2.3293613450365767) (0.47000000000000003, 2.3102941591596937) (0.48, 2.2911497128350162) (0.49, 2.271925956695621) (0.5, 2.2526207478964415) (0.51, 2.2332318439723697) (0.52, 2.213756896163325) (0.53, 2.1941934421482716) (0.54, 2.1745388981224876) (0.55, 2.1547905501435323) (0.56, 2.1349455446611034) (0.5700000000000001, 2.1150008781340093) (0.58, 2.0949533856235494) (0.59, 2.0747997282362727) (0.6, 2.0545363792699005)
};
\addplot[no marks, very thick, color=blue] coordinates{
(0.6, 2.054760863434977) (0.61, 2.0341966490258425) (0.62, 2.0143869319699843) (0.63, 1.9952922161116033) (0.64, 1.9768756464849293) (0.65, 1.9591027966731223) (0.66, 1.94194147603454) (0.67, 1.9253615546949996) (0.68, 1.9093348044518619) (0.6900000000000001, 1.8938347539510443) (0.7000000000000001, 1.8788365566859022) (0.71, 1.8643168705311137) (0.72, 1.8502537476684993) (0.73, 1.836626533887848) (0.74, 1.8234157763566685) (0.75, 1.8106031390503594) (0.76, 1.7981713251203177) (0.77, 1.7861040055534634) (0.78, 1.7743857535438372) (0.79, 1.7630019840564422) (0.8, 1.751938898116266) (0.81, 1.7411834314023216) (0.8200000000000001, 1.7307232067682339) (0.8300000000000001, 1.720546490348065) (0.84, 1.710642150939191) (0.85, 1.7009996223836499) (0.86, 1.69160886869583) (0.87, 1.6824603517080718) (0.88, 1.6735450010269943) (0.89, 1.6648541861124209) (0.9, 1.656379690307903) (0.91, 1.6481136866672679) (0.92, 1.6400487154354744) (0.93, 1.6321776630545954) (0.9400000000000001, 1.6244937425770176) (0.9500000000000001, 1.6169904753781783) (0.96, 1.6096616740703682) (0.97, 1.6025014265275033) (0.98, 1.5955040809383239) (0.99, 1.5886642318123736) (1.0, 1.5819767068693265)
};
\end{axis}
\end{tikzpicture}

\vspace*{.3cm}

\begin{tikzpicture}
\begin{axis}[
        width=\graphwidth, 
        height=\graphheight,
        axis line style={thick},
        xlabel={(c)},
        ymin=1, ymax=3.3,
        axis x line={bottom},
        axis y line={left},
        xtick distance=0.5,
        x axis line style=-,
    ]
\addplot[no marks, very thick, color=green!50!black] coordinates {
(0.0, 1.0) (0.01, 1.0101010101010102) (0.02, 1.0204081632653061) (0.03, 1.0309278350515465) (0.04, 1.0416666666666667) (0.05, 1.0526315789473684) (0.06, 1.0638297872340425) (0.07, 1.0752688172043012) (0.08, 1.0869565217391304) (0.09, 1.0989010989010988) (0.1, 1.1111111111111112) (0.11, 1.1235955056179776) (0.12, 1.1363636363636365) (0.13, 1.1494252873563218) (0.14, 1.1627906976744187) (0.15, 1.1764705882352942) (0.16, 1.1904761904761905) (0.17, 1.2048192771084338) (0.18, 1.2195121951219512) (0.19, 1.2345679012345678) (0.2, 1.25) (0.21, 1.2658227848101264) (0.22, 1.282051282051282) (0.23, 1.2987012987012987) (0.24, 1.3157894736842106) (0.25, 1.3333333333333333) (0.26, 1.3513513513513513) (0.27, 1.36986301369863) (0.28, 1.3888888888888888) (0.29, 1.4084507042253522) (0.3, 1.4285714285714286) (0.31, 1.4492753623188408) (0.32, 1.4705882352941178) (0.33, 1.492537313432836) (0.34, 1.5151515151515154) (0.35000000000000003, 1.5384615384615388) (0.36, 1.5625) (0.37, 1.5873015873015872) (0.38, 1.6129032258064517) (0.39, 1.639344262295082) (0.4, 1.6666666666666667) (0.41000000000000003, 1.6949152542372883) (0.42, 1.7241379310344827) (0.43, 1.7543859649122806) (0.44, 1.7857142857142856) (0.45, 1.8181818181818181) (0.46, 1.8518518518518516) (0.47000000000000003, 1.8867924528301885) (0.48, 1.923076923076923) (0.49, 1.9607843137254901) (0.5, 2.0) (0.51, 2.0408163265306123) (0.52, 2.0833333333333335) (0.53, 2.127659574468085) (0.54, 2.173913043478261) (0.55, 2.2222222222222223) (0.56, 2.272727272727273) (0.5700000000000001, 2.3255813953488373) 
};
\addplot[no marks, very thick, color=green!50!black, dashed] coordinates {
(0.58, 2.380952380952381) (0.59, 2.4390243902439024) (0.6, 2.5) (0.61, 2.564102564102564)
};
\addplot[no marks, very thick, color=orange] coordinates {
(0.47000000000000003, 2.3102941591596937) (0.48, 2.2911497128350162) (0.49, 2.271925956695621) (0.5, 2.2526207478964415) (0.51, 2.2332318439723697) (0.52, 2.213756896163325) (0.53, 2.1941934421482716) (0.54, 2.1745388981224876) (0.55, 2.1547905501435323) (0.56, 2.1349455446611034) (0.5700000000000001, 2.1150008781340093) (0.58, 2.0949533856235494) (0.59, 2.0747997282362727) (0.6, 2.0545363792699005)
};
\addplot[no marks, dashed, very thick, color=orange] coordinates {
(0.4, 2.4422463166925557) (0.41000000000000003, 2.4236042300649854) (0.42, 2.4048955047559195) (0.43, 2.386118547331526) (0.44, 2.3672716990495153) (0.45, 2.3483532320147145) 
};
\addplot[no marks, very thick, color=blue] coordinates{
(0.6, 2.054760863434977) (0.61, 2.0341966490258425) (0.62, 2.0143869319699843) (0.63, 1.9952922161116033) (0.64, 1.9768756464849293) (0.65, 1.9591027966731223) (0.66, 1.94194147603454) (0.67, 1.9253615546949996) (0.68, 1.9093348044518619) (0.6900000000000001, 1.8938347539510443) (0.7000000000000001, 1.8788365566859022) (0.71, 1.8643168705311137) (0.72, 1.8502537476684993) (0.73, 1.836626533887848) (0.74, 1.8234157763566685) (0.75, 1.8106031390503594) (0.76, 1.7981713251203177) (0.77, 1.7861040055534634) (0.78, 1.7743857535438372) (0.79, 1.7630019840564422) (0.8, 1.751938898116266) (0.81, 1.7411834314023216) (0.8200000000000001, 1.7307232067682339) (0.8300000000000001, 1.720546490348065) (0.84, 1.710642150939191) (0.85, 1.7009996223836499) (0.86, 1.69160886869583) (0.87, 1.6824603517080718) (0.88, 1.6735450010269943) (0.89, 1.6648541861124209) (0.9, 1.656379690307903) (0.91, 1.6481136866672679) (0.92, 1.6400487154354744) (0.93, 1.6321776630545954) (0.9400000000000001, 1.6244937425770176) (0.9500000000000001, 1.6169904753781783) (0.96, 1.6096616740703682) (0.97, 1.6025014265275033) (0.98, 1.5955040809383239) (0.99, 1.5886642318123736) (1.0, 1.5819767068693265)
};
\end{axis}
\end{tikzpicture}
\hspace*{\graphinbetween}
%
\hspace*{4.5pt}
\begin{tikzpicture}
\begin{axis}[
        width=\graphwidth, 
        height=\graphheight,
        axis line style={thick},
        xlabel={(d)},
        ymin=1, ymax=3.3,
        axis x line={bottom},
        axis y line={left},
        xtick distance=0.5,
        x axis line style=-
    ]
\addplot[no marks, very thick, color=green!50!black] coordinates {
(0.0, 1.0) (0.01, 1.0101010101010102) (0.02, 1.0204081632653061) (0.03, 1.0309278350515465) (0.04, 1.0416666666666667) (0.05, 1.0526315789473684) (0.06, 1.0638297872340425) (0.07, 1.0752688172043012) (0.08, 1.0869565217391304) (0.09, 1.0989010989010988) (0.1, 1.1111111111111112) (0.11, 1.1235955056179776) (0.12, 1.1363636363636365) (0.13, 1.1494252873563218) (0.14, 1.1627906976744187) (0.15, 1.1764705882352942) (0.16, 1.1904761904761905) (0.17, 1.2048192771084338) (0.18, 1.2195121951219512) (0.19, 1.2345679012345678) (0.2, 1.25) (0.21, 1.2658227848101264) (0.22, 1.282051282051282) (0.23, 1.2987012987012987) (0.24, 1.3157894736842106) (0.25, 1.3333333333333333) (0.26, 1.3513513513513513) (0.27, 1.36986301369863) (0.28, 1.3888888888888888) (0.29, 1.4084507042253522) (0.3, 1.4285714285714286) (0.31, 1.4492753623188408) (0.32, 1.4705882352941178) (0.33, 1.492537313432836) (0.34, 1.5151515151515154) (0.35000000000000003, 1.5384615384615388) (0.36, 1.5625) (0.37, 1.5873015873015872) (0.38, 1.6129032258064517) (0.39, 1.639344262295082) (0.4, 1.6666666666666667) (0.41000000000000003, 1.6949152542372883) (0.42, 1.7241379310344827) (0.43, 1.7543859649122806) (0.44, 1.7857142857142856) (0.45, 1.8181818181818181) (0.46, 1.8518518518518516)};
\addplot[no marks, very thick, color=green!50!black, dashed] coordinates {
(0.47000000000000003, 1.8867924528301885) (0.48, 1.923076923076923) (0.49, 1.9607843137254901) (0.5, 2.0)
};
\addplot[no marks, very thick, color=orange] coordinates {
(0.22, 1.7662966631854253) (0.23, 1.7413696144591533) (0.24, 1.7172767225686623) (0.25, 1.693980668299439) (0.26, 1.6714464291006441) (0.27, 1.6496411079529687) (0.28, 1.628533777324831) (0.29000000000000004, 1.6080953366768282) (0.3, 1.5882983821497294) (0.31, 1.569117087224276) (0.32, 1.5505270932747737) (0.33, 1.5325054090555408) (0.33999999999999997, 1.5150303182619518) (0.35, 1.49808129439798) (0.36, 1.4816389222614654) (0.37, 1.4656848254281356) (0.38, 1.4502015991769779) (0.39, 1.4351727483538366) (0.4, 1.4205826297180204) (0.41000000000000003, 1.4064163983590119) (0.42000000000000004, 1.392659957807698) (0.43, 1.3792999134994852) (0.44, 1.3663235292756943) (0.45, 1.353718686635172) (0.45999999999999996, 1.341473846470489) (0.47, 1.3295780130427064) (0.48, 1.3180206999658048) (0.49, 1.30679189798667) (0.5, 1.2958820443592856)
};
\addplot[no marks, very thick, color=blue] coordinates {
(0.5, 1.2958820443592856) (0.51, 1.2939287867202471) (0.52, 1.2920430019739713) (0.53, 1.290218991842568) (0.54, 1.2884514418315076) (0.55, 1.2867354256194474) (0.56, 1.2850664021487017) (0.5700000000000001, 1.2834402064231802) (0.58, 1.2818530351846575) (0.59, 1.2803014286906234) (0.6, 1.2787822497750103) (0.61, 1.2772926612619266) (0.62, 1.2758301026489705) (0.63, 1.2743922668048309) (0.64, 1.2729770772545879) (0.65, 1.2715826664683203) (0.66, 1.2702073554318336) (0.67, 1.2688496346654494) (0.6799999999999999, 1.2675081467679596) (0.69, 1.2661816704958226) (0.7, 1.264869106339583) (0.71, 1.2635694635268904) (0.72, 1.262281848361096) (0.73, 1.2610054537933684) (0.74, 1.2597395501220068) (0.75, 1.2584834767132413) (0.76, 1.2572366346416284) (0.77, 1.2559984801540092) (0.78, 1.2547685188680577) (0.79, 1.2535463006239789) (0.8, 1.2523314149155491) (0.81, 1.251123486834146) (0.8200000000000001, 1.2499221734664119) (0.8300000000000001, 1.2487271606927515) (0.8400000000000001, 1.2475381603398377) (0.8500000000000001, 1.2463549076457152) (0.86, 1.2451771590009595) (0.87, 1.2440046899336994) (0.88, 1.2428372933101528) (0.89, 1.241674777725758) (0.9, 1.2405169660649682) (0.91, 1.2393636942104544) (0.9199999999999999, 1.2382148098847459) (0.9299999999999999, 1.237070171609424) (0.94, 1.235929647768737) (0.95, 1.2347931157660827) (0.96, 1.2336604612631867) (0.97, 1.2325315774929906) (0.98, 1.2314063646383295) (0.99, 1.2302847292693855) (1.0, 1.2291665838337407)
};
\end{axis}
\end{tikzpicture}
\caption{\normalfont Overview of our upper bounds on the \poa\ ($y$-axis) for $\yFPA$ and $\yDPA$, respectively, as a function of $\gamma$ ($x$-axis).
(a) \mcce-\poa\ for multi-unit $\yFPA$ with overbidding (Theorem~\ref{thm:OBpoaMUA}).
(b) \mcce-\poa\ for multi-unit $\yFPA$ without overbidding (Theorems~\ref{thm:OBpoaMUA} \& \ref{thm:CCEpoaMUANOB}).
(c) \mcce-\poa\ for single-item $\yDPA$ without overbidding (Theorems \ref{thm:OBpoaMUA}, \ref{thm:CCEpoaMUANOB} \& \ref{thm:CCEsingleItemNpl}).
(d) \mcce-\poa\ for single-item $\yDPA$ without overbidding and $n = 2$ bidders (Theorems~\ref{thm:CCEsingleItemNpl}, \ref{thm:CCE2A} \& \ref{thm:CCE2B}).}
\label{fig:bound-overview}
\end{figure*}

This kind of bid rigging, where the winning bid ``magically'' aligns with the highest losing bid, is also known as \emph{magic number cheating} (see \cite{Ingraham2005}).
We refer the reader to \cite{LW2010,MM06} (and the references therein) for several other bid rigging examples.
Despite the fact that this form of corruption occurs frequently in practice, its negative impact is still poorly understood theoretically and only a few studies exist (mostly in the economics literature, see the related work section).

Our goal is to initiate the study of the social welfare loss caused by corrupt auctioneers in fundamental auction settings. We focus on a basic model that captures the magic number cheating mentioned above and generalizations thereof.
Clearly, more sophisticated bid rigging models are conceivable and we hope that our work will trigger some future work along these lines.

\myparagraph{Capturing Corruption with Hybrid Auctions.}
Consider the single-item auction setting and suppose the auctioneer runs a sealed bid first-price auction. After receipt of all bids, the auctioneer approaches the highest bidder with the offer that they can lower their bid to the second highest bid in exchange for a bribe. If the highest bidder agrees, they win the auction and pay the second-highest bid for the items plus the corresponding bribe to the auctioneer. If the highest bidder disagrees, they still win the auction but pay their bid for the item according to the first-price auction format. 
We assume that the bribe to be paid to the auctioneer is a pre-determined fraction $\gamma \in [0,1]$ 
of the savings of the highest bidder, i.e., the auctioneer's bribe amounts to $\gamma$ times the difference between the highest and second highest bid.
In case of the multi-unit auction setting, the procedure described above is adapted accordingly by offering the winning bidders to lower their bids to the highest losing bid. 

Observe that the payment scheme described above essentially reduces to the winning bidders paying a convex combination of $\gamma$ times their bids and $(1-\gamma)$ times the highest losing bid. As we will argue below, this setting is tantamount to studying a {\emph{hybrid auction (\yDPA)}}, where the items are assigned to the highest bidders (according to the respective single-item or multi-unit auction scheme) and the payments are a convex combination of the first-price and the second-price payments. By varying the parameter $\gamma \in [0, 1]$, \yDPA\ thus interpolates between the respective second-price auction ($\gamma = 0$) and the first-price auction ($\gamma = 1$) schemes.

More elaborate corruption schemes are of course conceivable. For example, the auctioneer might ask for a fixed amount rather than a fraction of the gains. Or, to avoid setting all bids to the magic number, the auctioneer may want to announce different (bribed) bids for every winning bidder. 
To capture more general corruption schemes, we also study what we term \emph{$\gamma$-approximate first-price auctions (\yFPA)} in this paper. Basically, these auctions implement a payment scheme that recovers at least a fraction of $\gamma \in [0,1]$ of the first-price payment rule (formalized below). 
The \yDPA\ also belongs to this class.
Not only does this capture more elaborate bribing schemes, it also handles the situation where some bidders have moral objections against partaking in such a scheme and do not accept the bribe. Additionally, this also enables us to capture corruption schemes with \emph{heterogeneous bidders}, i.e., where the auctioneer handles a different $\gamma_i$ for each bidder $i$.

In our view, the corruption settings described above serve as suitable motivations to analyze the resulting auctions \yDPA\ and \yFPA. But, at the same time, we feel that the study of such \emph{hybrid} auction formats is interesting in its own right, purely from an auction design perspective. For example, tight bounds on the price of anarchy (as a function of $\gamma$) provide insights on which payment rule should ideally be used to reduce the inefficiency.

\myparagraph{Our Contributions.} 
We study the inefficiency of equilibria of \yFPA\ and \yDPA, both in the single-item and the multi-unit auction setting.
More specifically, our goal is to obtain a precise understanding of the \emph{(robust) price of anarchy (POA)} \cite{KP99,R15,syrgkanis2013composable}. 
We opt for the price of anarchy notion here because it is one of the most appealing and widely accepted measures to assess the efficiency of equilibria, especially in the context of social welfare analysis.
We focus on the analysis of the robust price of anarchy under the complete information setting, incorporating equilibrium notions ranging from pure Nash equilibria (\pne) to coarse correlated equilibria (\mcce).\footnote{Several bounds are based on an adapted smoothness approach and extend to the incomplete information setting; see the extensions section below for more details.}
Moreover, we analyze the price of anarchy distinguishing between the case when bidders can overbid and when they cannot overbid their actual valuations for the items.

The main results that we obtain in this paper are summarized below (see Figure~\ref{fig:bound-overview} for an overview). 
Without any restrictions on the bids, we obtain the following result:
\begin{enumerate}
\item We prove an upper bound of $(1/\gamma) \cdot e^{1/\gamma}/(e^{1/\gamma} - 1)$ on the coarse correlated $\poa$ (\mcce-\poa) of any \yFPA\ in the multi-unit auction setting {when bidders can overbid}; see Figure~\ref{fig:bound-overview}(a).
Our upper bound follows from a suitable adaptation of the smoothness technique for multi-unit auctions \cite{syrgkanis2013composable,deKeijzer2013}. Further, by means of a single-item \yDPA, we prove a matching lower bound over the entire range $\gamma \in [0, 1]$.
{As a result, our bound settles the \mcce-\poa\ of \yFPA\ exactly for both the single-item and multi-unit auction setting over the entire range of $\gamma \in [0, 1]$.}
\end{enumerate} 

A standard assumption that often needs to be made to derive meaningful bounds on the POA is that the bidders cannot overbid (see also related work section).
Under the no-overbidding assumption, a more fine-grained landscape of the price of anarchy emerges:
\begin{enumerate}\setcounter{enumi}{1}
\item We show that the pure \poa\ (\pne-\poa) of \yDPA\ in the multi-unit auction setting is $1$ for $\gamma \in (0, 1)$. 
{This result is complemented by $\text{\pne-\poa} = 2.1885$ for $\gamma = 0$ \cite{BMTT19} and $\text{\pne-\poa} = 1$ for $\gamma = 1$ \cite{deKeijzer2013}.
Note that this reveals an interesting transition at $\gamma = 0$.}
\item We prove that the \mcce-\poa\ of any \yFPA\ in the multi-unit auction setting is upper bounded by 
\[
-(1-\gamma) \mathcal{W}_{-1}\left(-\frac{1}{e^{(2-\gamma)/(1-\gamma)}}\right),
\]
for $\gamma \lessapprox 0.607$ where $\mathcal{W}$ is the Lambert-$\mathcal{W}$ function. 
Combined with our upper bound (first contribution above) for $\gamma > 0.607$ (i.e. with overbidding), we obtain the combined bound depicted in Figure~\ref{fig:bound-overview}(b).
\item We prove that the correlated \poa\ (\mce-\poa) of \yDPA\ in the single-item auction setting is $1$ for every $\gamma \in (0, 1)$.
{This result together with $\text{\mce-\poa} = 1$ for $\gamma = 1$ \cite{feldman2016correlated} and our next result, shows that $\text{\mce-\poa} = 1$ for the entire range $\gamma \in [0, 1]$.}
\item We show that the \mcce-\poa\ of \yDPA\ in the single-item auction setting with $n$ bidders is bounded as indicated in Figure~\ref{fig:bound-overview}(c). {Concretely, we prove an upper bound of $1/(1-\gamma)$ and combine it with the multi-unit bounds from Figure~\ref{fig:bound-overview}(b)}. 
\item We show that the \mcce-\poa\ of \yDPA\ in the single-item auction setting with $n = 2$ bidders is bounded as indicated in Figure~\ref{fig:bound-overview}(d). 
This bound is derived by combining three different upper bounds, one of which the $1/(1-\gamma)$ bound from Figure~\ref{fig:bound-overview}(c). 
Technically, this is the most challenging part of the paper as we use the cumulative distribution functions (CDF) of equilibrium bids directly to derive these bounds. 
\end{enumerate}

\myparagraph{Implications.}
Altogether, our bounds provide a rather complete picture of the \poa\ of \yFPA\ and, in particular, \yDPA , for different equilibrium notions both in the single-item and the multi-unit auction setting and with and without overbidding. 
If the bidders can overbid then our (tight) bound on the \mcce-\poa\ (Figure~\ref{fig:bound-overview}(a)) shows that the \poa\ increases from a small constant $e/(e-1)$ to infinity as $\gamma$ decreases from $1$ to $0$. Thinking about \yDPA, we feel that this makes sense intuitively: As $\gamma$ approaches $0$, the auctioneer only withholds a small fraction of the surplus and the bidders are thus incentivized to exploit the corruption (as it comes at a low cost). In contrast, as $\gamma$ approaches $1$, the auctioneer charges a significant fraction of the surplus and while the bidders still have good reasons to join the corruption (explained below) they exploit it less drastically as it comes at a large cost.

Our bounds reveal that there is a substantial difference in the \poa\ depending on whether or not bidders can overbid; e.g., compare the bounds depicted in (a) and (b) (multi-unit setting), or (a) and (c) (single-item setting) in Figure~\ref{fig:bound-overview}. 
In general, it is not well-understood how the no-overbidding assumption influences the \poa\ of auctions; this question also relates to the \emph{price of undominated anarchy} studied by Feldman et al.~\cite{feldman2016correlated} (see related work below). 
Our bounds shed some light on this question for \yFPA. \\[-1ex]

\emph{Technical Merits.}
Our upper bounds for $\yFPA$ are based on an adapted smoothness notion which relates directly to the highest marginal winning bids (i.e., first-price payments). In particular, our smoothness argument does \emph{not} exploit the second-price payments of $\yDPA$ at all. As it turns out, this allows us to derive \emph{tight} bounds for $\yDPA$ and, more generally, for $\yFPA$ when bidders can overbid. On a high level, our results thus reveal that the (approximate) first-price payments are the determining component of such composed payment schemes. This triggers some interesting questions for future research.

In contrast, when overbidding is not allowed it becomes crucial to exploit the second-price payments of $\yDPA$ to obtain improved bounds. The price of anarchy of both the first-price auction and the second-price auctions is well understood in the single-item setting. However, it is not straightforward to extend these bounds to the combined payment scheme of \yDPA. In fact, to prove our bounds in Theorem \ref{thm:CCE2A} and Theorem \ref{thm:CCE2B}, we exploit constraints on the CDF of the first-price payments which are imposed by the CCE conditions;
but, additionally, we have to get a grip on the CDF of the second-price payments. We need several new insights (and a somewhat involved numerical analysis) to derive these bounds. 

\myparagraph{Extensions.}
Although we focus on the complete information setting in this paper, most of our bounds can be lifted to the incomplete information setting as introduced by Harsanyi \cite{Harsanyi67}, where players have private valuation functions drawn from a common prior.
Several of our upper bounds are based on an adapted smoothness approach for multi-unit auctions which extends (basically) directly to this incomplete information setting and (mixed) Bayes-Nash equilibria. 
All bounds displayed in Figure 1(a--c) remain valid for Bayes-Nash equilibria as well. 
Such smoothness-based extensions are by now rather standard.\footnote{More specifically, 
these extensions can be proven along similar lines of arguments as in \cite{deKeijzer2013}, where smoothness is used to bound the Bayes-Nash \poa\ of (standard) multi-unit auctions.}
Given that these extensions cause quite some notational overhead without adding much analytically, we defer further details to the full version of the paper.
  
\myparagraph{Related Work.}
There is a large body of research in economics studying collusion among bidders in auctions (see, e.g., \cite{GM87,AM92} for some standard references). Collusion between the auctioneer and the bidders in the form of bid rigging (as considered in this paper) has also been studied in the literature, but less intensively. Most existing works study certain aspects of equilibrium outcomes (e.g., equilibrium structure, auctioneer surplus, seller revenue, optimal bribe schemes, etc.); for an overview of the existing works along these lines, see \cite{LW00,MM06,LW2010} and the references therein.

The specific bid rigging model that we consider here was first studied by Menezes and Monteiro \cite{MM06} and a slight generalization thereof by Lengwiler and Wolfstetter \cite{LW00}, both for the single-item auction setting. These works consider a Bayesian setting where the valuations are independent draws from a common distribution function. Menezes and Monteiro \cite{MM06} prove the existence of symmetric equilibrium bidding strategies and derive an optimal bribe function for the auctioneer. The authors also study a \emph{fixed-price} bribe scheme, where the auctioneer charges a fixed amount that is independent of the gained surplus. 

Subsequently, Lengwiler and Wolfstetter \cite{LW2010} study a more complex bid rigging scheme for the single-item auction setting, where the auctioneer additionally offers the second highest bidder to increase their bid. To the best of our knowledge, none of the existing works studied the price of anarchy of corrupt auctions. 

Studying the price of anarchy in auctions has recently received a lot of attention; we refer to the survey paper by Roughgarden et al.~\cite{RST17} for an overview. A lot of work has gone into deriving bounds on the price of anarchy for various auction formats, both in the complete and incomplete information setting. The \emph{smoothness notion}, originally introduced by Roughgarden in \cite{R15} to analyze the robust price of anarchy of strategic games, turned out to be very useful in an auction context as well. Syrgkanis and Tardos \cite{syrgkanis2013composable} build upon this notion and provide a powerful (smoothness-based) toolbox for the analysis of a broad range of auctions that fall into their composition framework.

With respect to the multi-unit auction setting, de Keijzer et al.~\cite{deKeijzer2013} use an adapted smoothness approach to derive bounds on the \poa\ of Bayes-Nash equilibria for the first-price and the second-price multi-unit auction (mostly focussing on the setting with no overbidding). 
Our bounds coincide with theirs for the extreme points $\gamma = 0$ and $\gamma = 1$. For the more general class of subadditive valuations, the \poa\ of Bayes-Nash equilibria for the first-price multi-unit auction is 2, which follows from \cite{deKeijzer2013} and \cite{CKST16}. 
Birmpas et al.~\cite{BMTT19} recently settled the \pne-\poa\ of the second-price multi-unit auction and show that it is $2.1885$. 

Our bounds on the \mcce-\poa\ are also based on a smoothness approach. We use an adapted smoothness notion (inspired by \cite{deKeijzer2013,syrgkanis2013composable}) to derive our bounds, both in the overbidding and the no-overbidding setting. Interestingly, our smoothness proofs crucially exploit that the payments recover at least a faction of $\gamma$ of the first-price payments (but never exceed them). 
As a side result, Syrgkanis and Tardos \cite{syrgkanis2013composable} also derive a first bound on the \mce-\poa\ for \yDPA\ in the single-item auction setting; our bound (significantly) improves on theirs and exploits some additional ideas.

The \poa\ of the first-price and second-price auction has been investigated intensively for both the single-item and the multi-unit auction setting. An assumption that often needs to be made to derive meaningful bounds is that the bidders cannot overbid. 
For example, it is folklore that the \pne-\poa\ of the second-price single-item auction is unbounded if the bidders can overbid. On the other hand, it is one if bidders cannot overbid. In the second-price single-item auction, overbidding is a dominated strategy for each bidder and the no-overbidding assumption thus emerges naturally. But this might not be true in general. For example, for the second-price multi-unit auction, this analogy breaks already. 
We refer to \cite{feldman2013simultaneous} for a more general discussion of the no-overbidding assumption.

In general, the impact that the no-overbidding assumption has on the price of anarchy is not well-understood. This aspect also relates to the \emph{price of undominated anarchy} studied by Feldman et al.~\cite{feldman2016correlated}. The authors prove a clear separation for the \poa\ in single-item first-price auctions: While the \mce-\poa\ is $1$ (even with overbidding), the \mcce-\poa\ increases to $1.229$ (without overbidding) and $e/(e-1)$ (with overbidding).
A similar separation holds for the multi-unit auction setting and the uniform price auction, where the \pne-\poa\ is $(e-1)/e$ (without overbidding) \cite{MT15} and $2.1885$ (with overbidding) \cite{BMTT19}.
Our results contribute to this line of research also because we show that the \poa\ might improve significantly under the no-overbidding assumption.

\section{Preliminaries}
\label{sec:preliminaries}

\myparagraph{Standard Auction Formats.}
We focus on the description of the multi-unit auction setting; the single-item auction setting follows as a special case (choosing $k = 1$ below). 
In the multi-unit auction setting, there are $k \ge 1$ identical items (or goods) that we want to sell to $n \ge 2$ bidders (or players). We identify the set of bidders $N$ with $[n] = \set{1, \dots, n}$. 
Each bidder $i$ has a non-negative and non-decreasing valuation function $v_i: \set{0, \dots, k} \rightarrow \mathbb{R}_{\ge 0}$ with $v_i(0) = 0$, where $v_i(j)$ specifies $i$'s valuation for receiving $j$ items.
We assume that for each bidder $i \in N$ the valuation function $v_i$ is \emph{submodular} or, equivalently, that the marginal valuations are non-increasing, i.e., for every $j \in [k-1]$, $v_i(j) - v_i(j-1) \ge v_i(j+1) - v_i(j)$. 
The valuation function $v_i$ is assumed to be private information, i.e., it is only known to bidder $i$ themselves.
We use $\vals = (v_1, \dots, v_n)$ to denote the profile (or vector) of the valuation functions of the bidders. 
We assume that the bidders submit their bids according to the following \emph{standard format}: Each bidder $i$ submits a bid vector $\bids_i = (b_i(1),\dots, b_i(k))$ of $k$ non-negative and non-increasing \emph{marginal bids}, i.e., $b_i(j)$ specifies the additional amount $i$ is willing to pay for receiving $j$ instead of $j-1$ items. The overall amount that $i$ bids for receiving $q$ items is thus $\sum_{j = 1}^q b_i(j)$. {For $k=1$ we write $\bids_i = b_i(1)$.}

Consider a multi-unit auction setting and suppose the auctioneer uses an \emph{auction mechanism $\mech$} to determine an assignment of the items and the respective payments of the bidders. 
Each bidder submits their bid vector $\bids_i$ to the mechanism. Based on the bidding profile $\bids = (\bids_1, \dots, \bids_n)$, the mechanism $\mech$ orders the submitted marginal bids non-increasingly (breaking ties in an arbitrary but consistent way) and assigns the $k$ items to the bidders who submitted the $k$ highest marginal bids (according to this order). {We use $\beta_j(\bids)$ to refer to the $j$-th lowest winning (marginal) bid in $\bids$, i.e., $\beta_k(\bids) \geq \ldots \geq \beta_1(\bids)$.}
We use $\alloc(\bids) = (x_1(\bids), \dots, x_n(\bids))$ to refer to the resulting allocation, where $x_i(\bids)$ specifies the number of items that bidder $i$ receives; $x_i(\bids) = 0$ if $i$ does not receive any item.
Each bidder $i$ who receives at least one item is called a \emph{winner}.

There are two standard payment schemes that determine for each winner $i$ the respective payment $p_i(\bids)$; we adopt the convention that $p_i(\bids) = 0$ for each bidder $i$ who is not a winner.
\begin{itemize}
    \item \emph{First-price payment scheme}: Every bidder $i$ pays their bid for the received items, i.e., $p_i(\bids) = \sum_{j=1}^{x_i(\bids)} b_i(j)$
    \item \emph{Second-price payment scheme}: Every bidder $i$ pays the highest losing bid $\bar{p}(\bids)$ for each received item, i.e., $p_i(\bids) = x_i(\bids) \bar{p}(\bids)$
\end{itemize}
Suppose we fix the payment scheme of mechanism $\mech$ according to one of these schemes. We refer to mechanism $\mech$ with the first-price payment or the second-price payment scheme, respectively, as \emph{$\fpa$} or \emph{$\spa$}.\footnote{We remark that in the multi-unit auction setting these auctions are usually referred to as \emph{discriminatory price auction} and \emph{uniform price auction}; however, here we stick to the given naming convention to align it with the common terminology of the single-item auction setting.}

The utility $u_i^{v_i}(\bids)$ of bidder $i$ is defined as the total valuation minus the payment for receiving $x_i(\bids)$ items, i.e., $u^{v_i}_i(\bids) = v_i(x_i(\bids)) - p_i(\bids)$; note that $u^{v_i}_i(\bids) = 0$ by definition if bidder $i$ is not a winner. 
Whenever $v_i$ is clear from the context, we simply denote the utility of bidder $i$ by $u_i(\bids)$.
We assume that each bidder strives to maximize their utility. 

Finally, we introduce some standard assumptions that we use throughout this paper; we adopt the convention that the first two must always be satisfied by a mechanism. 
\begin{enumerate}    
    \item 
    \emph{No positive transfers (NPT):} The payment of each bidder $i$ is non-negative, i.e., $p_i(\bids) \ge 0$.
   \item 
\emph{Individual rationality (IR):} The payment of each bidder $i$ does not exceed their bid, i.e., $p_i(\bids) \le \sum_{j = 1}^{x_i(\bids)} b_{i}(j)$. 
  \item 
\emph{No overbidding (\nob):} The bid vector of each bidder $i$ does not exceed their valuations, i.e., for every $q \in [k]$, $\sum_{j=1}^q b_i(j) \allowbreak \le v_i(q)$.
\end{enumerate}

\myparagraph{Approximate First-Price Auctions.}
In this paper, we also consider auctions with first-price approximate payment schemes. The allocation is still determined as above, but the payment scheme is relaxed as follows: We say that a mechanism $\mech$ with payment rule $\pays = (p_1(\bids), \dots, p_n(\bids))$ is a \emph{$\gamma$-approximate first-price auction (\yFPA)} for some $\gamma \in [0, 1]$ if it always recovers at least a fraction of $\gamma$ of the first-price payments, i.e., for every bidding profile $\bids$, $\sum_{i \in N} p_i(\bids) \ge \gamma \sum_{j=1}^k \beta_j(\bids)$. 
Further, if for every bidding profile $\bids$ it holds that $\sum_{i\in N} p_i(\bids) \leq \sum_{j=1}^k \beta_j(\bids)$ then we call the mechanism \emph{first-price dominated}.
Note that every mechanism that satisfies individual rationality must be first-price dominated.

\myparagraph{Equilibrium Notions and the Price of Anarchy.}
We focus on the complete information setting here. Below, we briefly review the different equilibrium notions used in this paper. 
A bidding profile $\bids = (\bids_1, \dots, \bids_n)$ is a \emph{pure Nash equilibrium (PNE)} if no bidder has an incentive to deviate unilaterally; more formally, $\bids$ is a PNE if for every bidder $i$ and every bidding profile $\bids'_i$ of $i$ it holds that $u_i(\bids) \ge u_i(\bids'_i, \bids_{-i})$. 
Here we use the standard notation $\bids_{-i}$ to refer to the bid vector $\bids$ with the $i$th component being removed; $(\bids'_i, \bids_{-i})$ then refers to the bid vector $\bids$ with the $i$th component being replaced by $\bids'_i$. 

We also consider randomized bid vectors. Suppose bidder $i$ chooses their bid vectors randomly according to a probability distribution $\Bids_i$, independently of the other bidders. Let $\Bids = \prod_{i \in N} \Bids_i$ be the respective product distribution. Then $\Bids$ is a \emph{mixed Nash equilibrium (MNE)} if for every bidder $i$ and every bid vector $\bids'_i$ it holds that $\E_{\bids \sim \Bids} [u_i(\bids)] \ge \E_{\bids \sim \Bids} [u_i(\bids'_i, \bids_{-i})]$.
We may also allow correlation among the bidders. Let $\Bids$ be a joint distribution over bidding profiles of the bidders. 
Then $\Bids$ is a \emph{correlated equilibrium (CE)} if for every bidder $i \in N$ and for every {deviation function} $m_i(\bids_i)$ it holds that $\E_{\bids \sim \Bids} [u_i(\bids)] \ge \E_{\bids \sim \Bids} [u_i(m_i(\bids_i), \bids_{-i})]$. Intuitively, conditional on bid vector $\bids_i$ being realized, $i$ has no incentive to deviate to any other bid vector $m_i(\bids_i)$.
The most general equilibrium notion that we consider in this paper is defined as follows: 
Let $\Bids$ be a joint distribution over bidding profiles of the bidders. Then $\Bids$ is a \emph{coarse correlated equilibrium (CCE)} if for every bidder $i$ and every bid vector $\bids'_i$ it holds that $\E_{\bids \sim \Bids} [u_i(\bids)] \ge \E_{\bids \sim \Bids} [u_i(\bids'_i, \bids_{-i})]$.
Below, we also use $\pne(\vals)$, $\mne(\vals)$, $\mce(\vals)$ and $\mcce(\vals)$ to refer to the sets of pure, mixed, correlated and coarse correlated equilibria with respect to a valuation profile $\vals = (v_1, \dots, v_n)$, respectively. 

We define the \emph{social welfare} of a bidding profile $\bids = (\bids_1, \dots, \bids_n)$ as the overall valuation obtained by the bidders, i.e., $\SW(\bids) = \sum_{i \in N} v_i(x_i(\bids))$. 
Note that although social welfare is defined independently of the payments, we can equivalently write $SW(\bids) = \sum_{i\in N} u_i(\bids) + p_i(\bids)$.
The expected social welfare of a joint distribution $\Bids$ over bidding profiles is then defined as $\pe{\SW(\Bids)} = \E_{\bids \sim \Bids} [\SW(\bids)]$.
We use $\opt(\vals)$ to refer to an assignment that maximizes the social welfare with respect to the valuation functions $\vals = (v_1, \dots v_n)$; i.e., $\SW(\opt(\vals)) = \sum_{i \in N} v_i(x_i(\vals))$ is the maximum social welfare achievable for the bidders. The assignment $\opt(\vals)$ is also called a \emph{social optimum}.

The \emph{price of anarchy} is defined as the maximum ratio of the social welfare of the social optimum and the (expected) social welfare of an equilibrium. 
Let $X$ be a placeholder that refers to one of the equilibrium notions above, i.e., $X \in \set{\pne, \mne, \mce, \mcce}$. More formally, given a valuation profile $\vals = (v_1, \dots, v_n)$, the \emph{price of anarchy} with respect to $X$ (or \emph{$X$-\poa} for short) is defined as $\text{$X$-\poa}(\vals) = \sup_{\Bids \in X(\vals)}{\SW(\opt(\vals))}/\pe{\SW(\Bids)}$. 
The price of anarchy of an auction format then refers to the worst-case price of anarchy over all possible valuation profiles, i.e., $\text{$X$-\poa} = \sup_{\vals} \text{$X$-}\poa(\vals)$. 
We use \emph{\pne-\poa}, \emph{\mne-\poa}, \emph{\mce-\poa} and \emph{\mcce-\poa} to refer to the respective price of anarchy notions.

\section{Capturing Corruption with \yFPA}

We give a formal description of the model that we consider 
and elaborate on its relation to the \yhybrid. 
We also introduce the adapted smoothness approach.

\myparagraph{Corruption in Auctions.}
Suppose the bidders submit their bid vectors $\bids = (\bids_1, \dots, \bids_n)$ in a ``sealed manner'', i.e., at first only the auctioneer sees the bidding profile $\bids$.\footnote{It is important to realize though that the final bids, which might not necessarily correspond to the submitted ones, might have to be revealed eventually because the bidders might want to verify the ``soundness'' of the outcome of the auction.}
After receipt of the bidding profile $\bids$, the auctioneer runs a first-price multi-unit auction (see Section~\ref{sec:preliminaries}) to obtain the respective assignment $\alloc(\bids) = (x_1(\bids), \dots, x_n(\bids))$ and payments $\pays(\bids) = (p_1(\bids), \dots, p_n(\bids))$ but does not reveal this outcome yet. 
The auctioneer then approaches each winning bidder $i$ individually with the offer that they can lower all their $x_i(\bids)$ winning bids to the highest losing bid $\bar{p}(\bids)$ (while receiving the same number of items), in exchange for a fixed fraction $\gamma \in [0,1]$ of the surplus gained by $i$. The bidder can either reject or accept this offer. 
If bidder $i$ rejects the offer, the allocation $x_i(\bids)$ and respective payment $p_i(\bids)$ remain unmodified. 
If bidder $i$ accepts the offer, they receive the $x_i(\bids)$ items at a reduced price of $\bar{p}(\bids)$ each, but additionally pay a fee $\fee^\gamma_i$ of $\gamma$ times the surplus to the auctioneer; more formally, the total payment of a winning bidder $i$ who accepts the offer is 
\[
p^\gamma_i(\bids) = x_i(\bids) \bar{p}(\bids) + \fee^\gamma_i(\bids)
\enspace\text{where}\enspace
\fee^\gamma_i(\bids) = \gamma \sum_{j = 1}^{x_i(\bids)} (b_i(j) - \bar{p}(\bids)).
\]
We also refer to this setting as the \emph{\ycfpa}.\footnote{%
As the final payments are dependent on $\gamma$, we (implicitly) assume that the bidders are aware of this parameter when considering the complete information setting here (much alike it is assumed that the bidders know the used payment scheme in other auction formats). 
}

Note that the change in the bid vector of player $i$ conforms to 
the imposed bidding format, i.e., the modified marginal bids of bidder $i$ are still non-negative and non-increasing. 
It is not hard to show that it is a dominant strategy for every winning bidder to accept the offer of the auctioneer, independently of the parameter $\gamma$ (see Appendix~\ref{sec:AcceptOffer}).
Subsequently, we assume that each winning bidder always accepts the offer.

\begin{toappendix}

\label{app:model}
\subsection{Dominant Strategy to Accept Offer} \label{sec:AcceptOffer}

We show that it is a dominant strategy for every winning bidder to always accept the offer of the auctioneer (independent of $\gamma$). 

\begin{proposition}\label{ass:always-accept}
Fix some $\gamma \in [0,1]$ and consider a \newline\ycfpa. 
We can assume without loss of generality that each winning bidder always accepts the offer of the auctioneer. 
\end{proposition}
\begin{proof}
Observe that the total payment to be made by a winning bidder $i$ who accepts the offer becomes 
\begin{equation}\label{eq:hybrid-payment:app}
p^\gamma_i(\bids) 
= x_i(\bids) \bar{p}(\bids) + \fee^\gamma_i(\bids) 
= \gamma \sum_{j = 1}^{x_i(\bids)} b_i(j) + (1-\gamma) x_i(\bids) \bar{p}(\bids).
\end{equation}
Clearly, each winning bid $j$ of $i$ satisfies $b_i(j) \ge \bar{p}(\bids)$. 
Thus, $p^\gamma_i(\bids) \le \sum_{j = 1}^{x_i(\bids)} b_i(j) = p_i(\bids)$, where $p_i(\bids)$ is the payment that $i$ would have to pay when rejecting the offer. 
In fact, this inequality is strict unless all winning bids of $i$ are equal to $\bar{p}(\bids)$ or $\gamma = 1$. In both these cases, the offer made by the auctioneer does not have any effect for $i$ (as there is no surplus generated in the former case, and no difference in the final payment of $i$ in the latter case). 
Said differently, each winning bidder can only benefit from accepting the offer. Observe also that the above arguments hold for every winning bidder independently of what the other bidders do. Further, the final allocation remains invariant (assuming a consistent tie breaking rule).
We conclude that it is a dominant strategy for every winning bidder to accept the offer of the auctioneer. 
\qed
\end{proof}

\end{toappendix}

\myparagraph{Hybrid Auction Scheme.}
We introduce our novel hybrid auction scheme, which we term \emph{\yhybrid} (or \emph{$\yDPA$} for short): 
\yDPA\ uses the same allocation rule as in the multi-unit auction setting (see Section~\ref{sec:preliminaries}), but uses a convex combination of the first-price and second-price payment scheme (parameterized by $\gamma$), i.e., 
\begin{equation}\label{eq:hybrid-payment}
p^\gamma_i(\bids) 
= \gamma \sum_{j = 1}^{x_i(\bids)} b_i(j) + (1-\gamma) x_i(\bids) \bar{p}(\bids).
\end{equation}
Said differently, \yDPA\ interpolates between \spa\ ($\gamma = 0$) and \fpa\ ($\gamma = 1$) as $\gamma$ varies from $0$ to $1$. 
{It is immediate that every \yDPA\ is a \yFPA.}
We also use $p^\gamma(\bids)$ to refer to the above payment in the single-item auction setting.

The following proposition follows immediately from the discussion above and allows us to focus on the \poa\ of \yDPA\ to study \ycfpa s.

\begin{proposition}\label{prop:equivalence}
Fix some $\gamma \in [0,1]$. Then the $\ycfpa$ and $\yDPA$ admit the same set of equilibria and have identical social welfare objectives. Therefore, the price of anarchy for both these settings is the same. 
\end{proposition}

\myparagraph{Other Corruption Models.}
In our basic bid rigging model introduced above all winning bidders lower their bids to the highest losing bid. While this magic number bidding phenomenon has been observed in real-life for single-item auctions (as mentioned in the introduction), it might seem somewhat awkward in the multi-unit auction setting.\footnote{We refer to Appendix~\ref{sec:altMagicNumber} for further discussion on how the auctioneer could ``camouflage'' the magic number bidding in this case.}
We therefore consider more general corruption schemes that also capture non-uniform bid rigging. More precisely, most of our upper bounds hold for the more general class of \yFPA\ introduced above. These auctions capture several additional corruption settings.
For example, suppose some bidders never accept the offer of the auctioneer (say due to moral objections) and their payments thus remains the first-price payment. While this setting is not covered by \yDPA, it is covered by \yFPA. As another example, if the auctioneer handles a different fraction $\gamma_i$ for each bidder $i$, the resulting auction is $\gamma$-FPA with $\gamma = \min_{i \in N} \gamma_i$.

\begin{toappendix}
\subsection{Equivalence for Non-Uniform Bid Rigging}
\label{sec:altMagicNumber}

We exemplify the implications of our bounds on the \poa\ of \newline\yDPA\ for an alternative, non-uniform bid rigging model.

As before, the bidders submit their bid vectors $\bids = (\bids_1, \dots, \bids_n)$ to the auctioneer who runs a first-price multi-unit auction. The auctioneer then approaches each winning bidder $i$ individually with the offer that they can lower their $x_i(\bids)$ winning bids. However, in contrast to the basic model, the auctioneer and bidder $i$ agree to ``camouflage'' their bid rigging by bidding the highest losing bid $\bar{p}(\bids)$ plus a fraction $\alpha \in [0,1]$ of the surplus $b_i(j) - \bar{p}(\bids)$ for each $j \in [x_i(\bids)]$. Note that this maintains the relative order among the winning bids and the magic number cheating becomes less obvious (as the winning bids fluctuate more). The remaining surplus of $(1-\alpha) (b_i(j) - \bar{p}(\bids))$ is then split, where the auctioneer withholds a fraction of $\beta \in [0,1]$. As before, bidder $i$ can either reject or accept the offer. But, also here, it is not hard to see that accepting the offer is a dominant strategy. 
The total payment of a winning bidder $i$ is then
\begin{align*}
& p^{(\alpha,\ \beta)}_i(\bids) 
= \sum_{j = 1}^{x_i(\bids)} (\bar{p}(\bids) + \alpha (b_i(j) - \bar{p}(\bids))) + \fee^{(\alpha,\ \beta)}_i(\bids), \quad \text{where} \\
& \qquad \fee^{(\alpha,\ \beta)}_i(\bids)  = \beta \sum_{j = 1}^{x_i(\bids)} (1-\alpha) (b_i(j) - \bar{p}(\bids)).
\end{align*}
After simplifying, we obtain
\begin{align*}
p^{(\alpha,\ \beta)}_i(\bids) 
& = (\alpha + \beta(1-\alpha)) \sum_{j = 1}^{x_i(\bids)} b_i(j) 
+ (1 -\alpha -\beta(1-\alpha)) x_i(\bids) \bar{p}(\bids).
\end{align*}
If we define $\gamma = \alpha + \beta - \alpha\beta$, the above payments $p^{(\alpha,\ \beta)}_i$ are equivalent to $p_i^\gamma$ as defined in \eqref{eq:hybrid-payment}. Note also that this mapping satisfies $\gamma \in [0,1]$ for every $\alpha, \beta \in [0,1]$. 
Said differently, given $\alpha, \beta \in [0,1]$ the price of anarchy of the above non-uniform bid rigging scheme is determined by the price of anarchy of \yDPA\ with $\gamma = \alpha + \beta - \alpha\beta$.

\end{toappendix}

\subsection{Adapted Smoothness Notion} \label{secSmoothNotion}

We introduce our adapted smoothness notion (based on the ones given in \cite{syrgkanis2013composable,deKeijzer2013}) to derive upper bounds on the coarse correlated price of anarchy of \yhybrid.

Given a bidding profile $\bids$, we let $\beta_j(\bids)$ refer to the $j$th lowest winning bid under $\bids$. 
\begin{definition}\label{apx:def:weak-smooth-new}
A mechanism $\mathcal{M}$ for the multi-unit auction setting is \emph{$(\lambda, \mu)$-smooth}  for some $\lambda > 0$ and $\mu \ge 0$ if 
for every valuation profile $\vals$ and 
for each bidder $i \in N$
there exists a (possibly randomized) deviation $\Bids'_i$
such that for every bidding profile $\bids$ we have
\[
\sum_{i \in N} \pel[\bids_i' \sim \Bids_i']{u_i(\bids'_i, \bids_{-i})} 
\ge 
\lambda \SW(\opt(\vals)) - \mu \sum_{j=1}^{k} \beta_j(\bids).
\] 
\end{definition} 
In essence, this definition comes close to the \emph{weak smoothness} definition in \cite{syrgkanis2013composable}, but relates more directly to the winning bids in the multi-unit auction setting. A similar definition is also used in \cite{deKeijzer2013}, but there it is imposed on a per-player basis and used for the Bayesian setting.

\begin{theorem}\label{thm:smoothness-bound}
Let $\mathcal{M}$ be a \yFPA\ which is $(\lambda, \mu)$-smooth.
Then $\text{\mcce-\poa} \le \max\set{1, 1 + \mu - \gamma}/{\lambda}$, where we need that the no-overbidding assumption holds if $\mu > \gamma$.
\end{theorem}
\begin{proof}
Fix a valuation profile $\vals$ and let $\Bids$ be a coarse correlated equilibrium. 
Consider some player $i$ and let $\Bids'_i$ be the (randomized) deviation of bidder $i$ as given by the smoothness definition. Exploiting the coarse correlated equilibrium condition for $i$, we have for every (deterministic) bid vector $\bids'_i$ that $\pe[\bids \sim \Bids]{u_i(\bids)} \ge \pe[\bids \sim \Bids]{u_i(\bids'_i, \bids_{-i})}$ and thus also 
\begin{equation}\label{eq:deviateee}
\pe[\bids \sim \Bids]{u_i(\bids)} 
\ge \pe[\bids \sim \Bids]{\pe[\bids'_i \sim \Bids'_i]{u_i(\bids'_i, \bids_{-i})}}.
\end{equation}
Using this, we obtain 
\begin{align}
\pe{\SW(\Bids)}
& = \sum_{i \in N} \pel[\bids \sim \Bids]{u_i(\bids)  + p_i(\bids)}  \\
&\ge \sum_{i \in N} \pel[\bids \sim \Bids]{\pe[\bids'_i \sim \Bids'_i]{u_i(\bids'_i, \bids_{-i})} + p_i(\bids)} \notag \\
 & \ge \lambda \SW(\opt(\vals)) + ( \gamma - \mu) \pel[\bids \sim \Bids]{\sum_{j=1}^k \beta_j(\bids)} \label{eq:smooth-old},
\end{align}
where the first inequality follows from \eqref{eq:deviateee} and the second inequality holds because of the smoothness definition and because $\mech$ is first-price $\gamma$-approximate. 

We distinguish two cases: 

\underline{Case 1:} $\mu \le \gamma$. 
Using \eqref{eq:smooth-old}, we obtain 
$\pe{\SW(\Bids)} \ge  \lambda \SW(\opt(\vals))$
and thus 
$\poa(\vals) \le 1/\lambda$. 

\underline{Case 2:} $\mu > \gamma$.
Exploiting that the no-overbidding assumption holds in this case, we get that $\sum_{j=1}^k \beta_j(\bids) \le \sum_{i \in N} v_i(x_i(\bids))$. 
Using \eqref{eq:smooth-old}, we obtain
$
\pe{\SW(\Bids)}
\ge \lambda \SW(\opt(\vals)) + (\gamma-\mu) \pe{\SW(\Bids)}
$.
Rearranging terms yields
$
\poa(\vals) \le  (1+\mu-\gamma)/\lambda.
$
Combining both cases proves the claim. 
\qed
\end{proof}

We will use the above smoothness definition in combination with the following lemma, which we import from \cite{deKeijzer2013} (adapted to our setting). 
\begin{lemma}[Lemma 3 in  \cite{deKeijzer2013}]
\label{keylemma}
Let $\mech$ be a mechanism that is first-price dominated and let $\alpha > 0$ be fixed arbitrarily. Then for every valuation profile $\vals$ and for every bidder $i$ there exists a randomized deviation $\Bids'_i$ such that for every bidding profile $\bids$ we have 
\begin{equation}\label{eq:key}
\pe[\bids'_{i}\sim \Bids'_i]{u_i(\bids_i', \bids_{-i})} 
\geq 
\alpha \left(1-\frac{1}{e^{1/\alpha}}\right) v_i(\opt_i(\vals))  
- \alpha\sum_{j=1}^{\opt_i(\vals)} \beta_j(\bids)
\end{equation}
\end{lemma}

We can now prove Theorem~\ref{thm:poa-bound}.

\begin{theorem}\label{thm:poa-bound}
Let $\alpha > 0$ be fixed arbitrarily. The coarse correlated price of anarchy of any \yFPA\ is 
\begin{equation}\label{eq:start-from-here}
    \textit{\mcce-\poa} \le \frac{\max\set{1, 1 + \alpha - \gamma}}{\alpha (1- e^{-1/\alpha})},
\end{equation}
where we need that the no-overbidding assumption holds if $\alpha > \gamma$.
\end{theorem}
\begin{proof}
We use both Lemma~\ref{keylemma} and Theorem~\ref{thm:smoothness-bound}.

Note that $\sum_{i \in N} \sum_{j=1}^{\opt_i(\vals)} \beta_j(\bids) \le 
\sum_{j=1}^{k} \beta_j(\bids)$. Hence, by summing inequality \eqref{eq:key} over all players, we obtain that the mechanism is $(\alpha (1- e^{-1/\alpha}), \alpha)$-smooth. 
The claimed bound now follows from Theorem~\ref{thm:smoothness-bound}.
\qed
\end{proof}

\section{Overbidding} 
\label{sec:overbidding}

We derive a tight bound on the coarse correlated price of anarchy of \yFPA\ for $\gamma > 0$ in the multi-unit auction setting when bidders can overbid. Interestingly, tightness is already achieved by a  single-item \yDPA. It is known that the price of anarchy is unbounded for $\spa$ ($\gamma = 0$). The bound is displayed in Figure \ref{fig:bound-overview}(a). 

\begin{theorem} \label{thm:OBpoaMUA}
Consider a multi-unit \yFPA\ and suppose that bidders can overbid. For $\gamma \in (0, 1]$, the coarse correlated price of anarchy is 
    $
        \text{\mcce-\poa} \le \frac{1}{\gamma (1- e^{-1/\gamma})}.
    $
Further, this bound is tight, even for single-item \yDPA. 
\end{theorem}
\pagebreak 
\begin{proof}
\label{appendix:overbidding} 
\item 
\myparagraph{Upper bound:}
This bound is based on Theorem \ref{thm:poa-bound}. Since bidders can overbid in this setting, we restrict to the part of equation (\ref{eq:start-from-here}) that does not require the no-overbidding assumption, namely $\alpha \le \gamma$, with 
\[
\textit{\mcce-\poa} \leq \frac{1}{\alpha \left(1- e^{-1/\alpha}\right)}.
\]
To minimize this upper bound for any given $\gamma \in [0,1]$, consider its derivative with respect to $\alpha$, 
\begin{align*}
    - \frac{1}{\alpha^2 (1- e^{-1/\alpha})^2} \left(1- e^{-1/\alpha} - \alpha \frac{1}{\alpha^2} e^{-1/\alpha}\right) =  - \frac{1- (1+\tfrac{1}{\alpha}) e^{-1/\alpha}}{\alpha^2 (1- e^{-1/\alpha})^2}.  
\end{align*}
As 
$(1+1/\alpha) e^{-1/\alpha} < 1 $ for all $\alpha > 0$, the derivative is negative for all $\alpha > 0$. Therefore, the bound is minimized by maximizing $\alpha \in (0,\gamma]$. Substituting $\alpha = \gamma$ for any $\gamma \in (0,1]$ yields the upper bound. 

\item 
\myparagraph{Tight lower bound:} This bound can be proven to be tight for all $\gamma \in (0,1]$ by generalizing an example used by Syrgkanis \cite{phdthesisSyrgkanis} to provide a lower bound on the \mcce-\poa\ for the first-price single-item auction: 
Consider a single-item auction with two bidders and using the $\gamma$-hybrid pricing rule as defined above. We have $v_1 = v$ for some $v>0$ and $v_2=0$. If both bidders bid 0, the tie is broken in favor of bidder 2, whereas bidder 1 wins the auction if bidders tie with any positive bid. We construct a coarse correlated equilibrium for any $\gamma \in (0,1]$, with a welfare loss that matches the upper bound.

Let $t$ be a random variable with support $[0,(1-e^{-1/\gamma})v]$ whose cumulative distribution function (CDF) $F$ and density function $f$ {(which is well-defined for any $t \in (0,(1-e^{-1/\gamma})v]$)}, respectively, are given as 
\[
F(t)=(1-\gamma)+\frac{v}{v-t}\gamma e^{-1/\gamma}
\qquad
\text{and}
\qquad
f(t)=\frac{v}{(v-t)^2}\gamma e^{-1/\gamma}.
\]
Note that $F$ has an atom at $0$ with mass $(1-\gamma)+\gamma e^{-1/\gamma}$.

Consider a bidding profile $\Bids = (t,t)$.
Since ties are broken in favor of bidder 2 for $t=0$, they win with probability $(1-\gamma)+\gamma e^{-1/\gamma}$, which yields
\[
\frac{\SW(\opt(\vals))}{\E[\SW(\Bids)]} = \frac{v}{(1-F(0))v} = \frac{1}{1-(1-\gamma)-\gamma e^{-1/\gamma}}=\frac{1}{\gamma \left(1- e^{-1/\gamma}\right)}.
\]

It remains to show that $\Bids$ is a $\mcce$. For bidder 2, this is quite obvious, since they either win by bidding 0, or lose if $t>0$. Given any positive bid from bidder 1, the payment would be strictly greater than $v_2=0$, meaning bidder 2 could never profitably deviate. 

For bidder 1, we show that for any $\gamma \in (0,1]$, any deviation to a fixed bid $b_1=b$ with $b \in (0,(1-e^{-1/\gamma})v]$ leads to an expected utility of at most $ \E_{\bids \sim \Bids}[u_1(\bids)]$. To start with $\Bids$ itself, note that bidder 1 wins whenever $t>0$, and since both bidders bid $t$, we have a payment of $\gamma t + (1-\gamma)t = t$. Recalling that $v_1=v$, we get
\begin{align*}
    \E_{\bids \sim \Bids}[u_1(\bids)] &= \int_{0}^{(1-e^{-1/\gamma})v } (v - t)f(t) dt \\ 
    &= \int_{0}^{(1-e^{-1/\gamma})v } \frac{v}{v-t}\gamma e^{-1/\gamma}  dt \\
    &= v \gamma e^{-1/\gamma} \left[ -\ln (v-t)\right]_{0}^{(1-e^{-1/\gamma})v } \\ 
    &= v \gamma e^{-1/\gamma} \left( \ln (v) - \ln (e^{-1/\gamma}v)  \right) \\
    &= v \gamma e^{-1/\gamma}  \tfrac{1}{\gamma} = v e^{-1/\gamma}.
\end{align*}
If bidder 1 deviates to $b$, they win the item if $b\geq t$, and for each $t \in (0,b]$ bidder 1 pays $\gamma b + (1-\gamma) t$. Hence, the expected utility of bidder 1 becomes
\[
    \E_{t \sim F(t)}[u_1(b,t)] = \int_{0}^{b} (v - \gamma b - (1-\gamma) t) f(t) dt \\
\]
To facilitate the calculations, note that
\begin{align*}
    \int_{0}^{b} t f(t) dt &= \gamma v e^{-1/\gamma}  \int_{0}^{b} \frac{t}{(v-t)^2} dt \\ &= \gamma v e^{-1/\gamma} \left[ \frac{v}{v-t} + \ln(v-t) \right]_{0}^{b} \\
    &=  \left(\frac{v}{v-b} - 1 \right) \gamma v e^{-1/\gamma} + \ln \left(\frac{v-b}{v} \right) \gamma v e^{-1/\gamma} 
\end{align*}
and 
\[
    \int_{0}^{b} f(t) dt = F(b)-F(0) = \left( \frac{v}{v-b} -1\right) \gamma e^{-1/\gamma}.
\]
Using this, we get
\begin{align*}
    \E_{t \sim F(t)}[u_1(b,t)]  &= (v - \gamma b )\int_{0}^{b} f(t) dt - (1-\gamma) \int_{0}^{b} t f(t) dt \\
    &= (v - \gamma b -(1-\gamma)v) \left( \frac{v}{v-b} -1\right) \gamma e^{-1/\gamma} \\
    & \quad- (1-\gamma) \ln \left(\frac{v-b}{v} \right) \gamma v  e^{-1/\gamma} \\
    &= \gamma (v - b) \left( \frac{v}{v-b} -1\right) \gamma e^{-1/\gamma} \\
    &\quad- (1-\gamma) \ln \left(\frac{v-b}{v} \right) \gamma v e^{-1/\gamma} \\
    &= b \gamma^2 e^{-1/\gamma} - (1-\gamma) \ln \left(\frac{v-b}{v} \right) \gamma v e^{-1/\gamma}.
\end{align*}
Since $0<b<v$, note that $- \ln \left(\frac{v-b}{v} \right)$ is increasing in $b$. Since $\gamma \in (0,1]$, this implies the entire function above is increasing in $b$. Hence, it can be upper bounded by substituting the upper bound of the support: $b=(1-e^{-1/\gamma})v$. This yields
\begin{align*}
    \E_{t \sim F(t)}[u_1(b,t)] &\le (1-e^{-1/\gamma}) v \gamma^2 e^{-1/\gamma}  - (1-\gamma) \ln \left( e^{-1/\gamma} \right) \gamma v e^{-1/\gamma}  \\
    &= \left( (1-e^{-1/\gamma}) \gamma^2 + (1-\gamma) \right) v e^{-1/\gamma}\\
    &= \left( (1-e^{-1/\gamma}) \gamma^2 + (1-\gamma) \right)  \E_{\bids \sim \Bids}[u_1(\bids)].
\end{align*}
Therefore, $\E_{t \sim F(t)}[u_1(b,t)] \leq \E_{\bids \sim \Bids}[u_1(\bids)]$ for any $b \in (0,(1-e^{-1/\gamma})v]$ if
\begin{align*}
    (1-e^{-1/\gamma}) \gamma^2 + (1-\gamma)  \leq 1  \quad \iff \quad 
    \gamma (1-e^{-1/\gamma}) \leq 1
\end{align*}
which holds for any $\gamma \in (0,1]$ as required. This shows that bidder 1 does not have any profitable deviation in the interval $(0,(1-e^{-1/\gamma})v]$. 
Finally, since $b =(1-e^{-1/\gamma})v$ already gives $F(b)=1$, any higher bid will only lead to a (strictly) higher payment (since $\gamma>0$), thereby being (strictly) worse than bidding $b =(1-e^{-1/\gamma})v$. Hence, deviations to a bid higher than this upper bound of the support of $F(t)$ need not be considered. 

Concluding, $\Bids$ is a $\mcce$ for which the ratio of the social welfare of the social optimum and the expected social welfare of $\Bids$ exactly coincides with the upper bound derived in the previous section.  \qed
\end{proof}

\section{No Overbidding}
\label{sec:no-overbidding}

\subsection{Multi-Unit Auction}

In the previous section, we have completely settled the coarse correlated price of anarchy of \yFPA\ when overbidding is allowed. We see that especially when $\gamma$ gets small this has an extremely negative effect on the price of anarchy. In this section, we will investigate how these bounds improve under the no-overbidding assumption (\nob\ as defined above). 
It is a standard assumption to make and we will see that it leads to a significant improvement of the price of anarchy bounds, most notably for lower values of $\gamma$. 

We start with bounding the pure price of anarchy.

\subsubsection{Pure Price of Anarchy.}

The pure price of anarchy of \yDPA\ without overbidding has been analyzed before for $\gamma = 0$ and $\gamma = 1$: Birmpas et al. \cite{BMTT19} show that the \pne-\poa\ is $2.1885$ for the second-price multi-unit auction ($\gamma = 0$), while de Keijzer et al. \cite{deKeijzer2013} show that the \pne-\poa\ is 1 for the first-price multi-unit auction ($\gamma = 1$).
Interestingly, we do not find a smooth interpolation between these two boundary points when analyzing the \pne-\poa\ for the range $\gamma \in [0,1]$.
As it turns out, for \yDPA\ the \pne-\poa\ stays at 1 almost over the entire range, the only exception being at $\gamma = 0$ where it is $2.1885$ by the result of Birmpas et al. \cite{BMTT19}.

\begin{theorem}
Pure Nash equilibria of \yDPA\ without overbidding are always efficient, i.e., {\pne-\poa\ $ = 1$} for all $\gamma \in (0,1)$.
\end{theorem}

This theorem follows from a minor adaption of the result of de Keijzer et al. \cite{deKeijzer2013}, who show that pure Nash equilibria are always efficient for $\gamma=1$. Intuitively, the same result goes through for $\gamma > 0$, because when considering $\gamma > 0$ there is a first price component and thus a player always has an incentive to lower their winning bid to the highest losing bid as that would increase their utility. For more details we refer to Appendix \ref{appPNE}.

\begin{toappendix}
\label{appPNE}
De Keijzer et al. \cite{deKeijzer2013} prove the following theorem.
\begin{theorem} \label{thmKeijzer}
\cite{deKeijzer2013} Pure Nash equilibria of the Discriminatory Auction 
are always efficient, even for bidders with arbitrary valuation functions.
\end{theorem}
This theorem follows almost immediately from this lemma.
\begin{lemma} \label{lemKeijzer}
\cite{deKeijzer2013} Let $\bids$ be a pure Nash equilibrium in a given Discriminatory Auction where the bidders have general valuation functions. Let $d = \max\{b_i(j) : i\in [n], j \in [k], j > x_i(\bids)\}$. Then
\begin{enumerate} 
\item[(i)] For any bidder $i$ who wins at least one item under $\bids$, and for all $j\in [x_i(\bids)] : b_i(j) = d$,
\item[(ii)] $\ell d \leq \sum_{j = x_i(\bids) - \ell+1}^{x_i(\bids)} m_i(j)$, for all $i\in [n]$ and $\ell \in [x_i(\bids)]$,
\item[(ii)] $\sum_{j = x_i(\bids)+1}^{x_i(\bids)+\ell} m_i(j) \leq \ell d$, for all $i\in [n]$ and $\ell \in [k - x_i(\bids)]$,
\end{enumerate}
\end{lemma}
We can proceed along exactly the same line of arguments as in the proof of the lemma for $\gamma = 1$ to show that it holds for all $\gamma > 0$. Because there is a first-price component, the reasoning that a player would increase their utiltity by lowering their winning bid to the highest losing bid holds in the same way.
\end{toappendix}

\subsubsection{Coarse Correlated Price of Anarchy}

\begin{theorem} \label{thm:CCEpoaMUANOB}
Consider a multi-unit \yFPA\ and suppose that bidders cannot overbid. 
For $\gamma \lessapprox 0.607$, the coarse correlated price of anarchy is
\begin{equation}\label{eq:upperboundMUAlambert}
\text{\mcce-\poa} \le -(1-\gamma) \mathcal{W}_{-1}\left(-\frac{1}{e^{(2-\gamma)/(1-\gamma)}}\right).
\end{equation}
\end{theorem}

\sloppy

Combining the improved bound of Theorem~\ref{thm:CCEpoaMUANOB} with the bound of Theorem~\ref{thm:OBpoaMUA} yields the upper bound displayed in Figure \ref{fig:bound-overview}(b) for all $\gamma \in [0,1]$. In particular, we obtain $\textit{\mcce-\poa} \le -\mathcal{W}_{-1}(-e^{-2}) \approx 3.146$ for $\gamma=0$ and $\textit{\mcce-\poa} \le e/(e-1) \approx 1.582$ for $\gamma=1$. 

\fussy

\begin{proof}
Similar to the proof of Theorem \ref{thm:OBpoaMUA}, we choose some $\alpha>0$ to optimize the upper bound on the price of anarchy in Theorem \ref{thm:poa-bound} for any given $\gamma \in [0,1]$. 
As argued in the proof of Theorem~\ref{thm:OBpoaMUA}, it is optimal to use $\alpha = \gamma$ when restricting to $\alpha \leq \gamma$. 
Using the no-overbidding assumption, we can also set $\alpha \geq \gamma$ and obtain
\begin{equation}\label{eq:smooth-specW}
\textit{\mcce-\poa} \leq \frac{1 + \alpha - \gamma}{\alpha \left(1- e^{-1/\alpha}\right)}.
\end{equation}
This upper bound is minimized for
\begin{align} \label{eq:alphaLambert}
    \alpha = -\frac{1}{\mathcal{W}_{-1}\left(-e^{-(2-\gamma)/(1-\gamma)}\right)+ \frac{2-\gamma}{1-\gamma}},   
\end{align}
where $\mathcal{W}_{-1}$ is the lower branch of the Lambert $\mathcal{W}$ function. 
Substituting this into \eqref{eq:smooth-specW}, we obtain the upper bound in \eqref{eq:upperboundMUAlambert}. 
Importantly, the optimized bound in \eqref{eq:upperboundMUAlambert} is only valid if we have $\alpha \ge \gamma$, which does not hold for the entire range $\gamma \in [0,1]$ if we use (\ref{eq:alphaLambert}). 
More concretely, we have $\alpha \ge \gamma$ for all $\gamma \le 0.607...$ only. 
Thus, for $\gamma \leq 0.607...$ we can use (\ref{eq:upperboundMUAlambert}) to bound the price of anarchy. 
For $\gamma \geq 0.607...$ the best we can do is to choose $\alpha=\gamma$ and obtain the same $\textit{\mcce-\poa}$ bound as in Theorem \ref{thm:OBpoaMUA}. 
\qed
\end{proof}

\subsection{Single-Item \yDPA}

We can further improve the price of anarchy bounds for single-item \yDPA. It allows us to make more direct use of the payments giving us more control. We start with the general $n$-player setting, for which we show that the single-item \yDPA\ is fully efficient up to correlated equilibria. For coarse correlated equilibria, we then derive a strong bound for low values of $\gamma$, namely \textit{\mcce-\poa} $\le 1/(1-\gamma)$. This bound can in turn be complemented by the bound we derived for multi-unit auctions. Finally, to improve upon this multi-unit bound for the higher range of $\gamma$, we derive two technically more involved bounds that work specifically in a two-player setting.  

We need some more notation.
Given a bid vector $\bids$, let $\HB(\bids) = \max_{i}\bids_i$ and $\SB(\bids)$ denote the highest and second highest bid in $\bids$, respectively, and let $\HB_{-i}(\bids) = \max_{j \neq i} \bids_j$ be the highest bid excluding bid $\bids_i$. For a randomized bid vector $\Bids$, let $HB(\Bids)$ be the random variable equal to the highest bid when the bids are distributed according to $\Bids$. We sometimes write $\E[\HB(\Bids)]$ for $\E_{\bids \sim \Bids}[\HB(\bids)]$ (similarly for $\SB(\Bids)$ and $\HB_{-i}(\Bids)$).

\myparagraph{Correlated Price of Anarchy.}
We prove that \yDPA\ is fully efficient for all $\gamma \in (0,1]$ up to correlated equilibria. We extend a result in \cite{feldman2016correlated}, which only does it for $\gamma=1$. Below we show that for $\gamma=0$ even coarse correlated equilibria are always efficient, so that Theorem~\ref{thm:singleCE} in fact holds for all $\gamma \in [0,1]$.

\begin{toappendix} 
\subsection{Single-Item Auction} 
\end{toappendix}

\begin{theorem} \label{thm:singleCE}
Consider a single-item $\yDPA$ and suppose that bidders cannot overbid. 
Then, the correlated price of anarchy of $\yDPA$ is $1$ for all $\gamma\in (0,1]$.
\end{theorem}
\begin{proof}
Without loss of generality assume that player 1 has the highest valuation $v_1$. Assume towards contradiction that the \textit{\mce-\poa} is not 1. Then, there must be a player $i$ for which $v_i < v_1$ who has a positive probability of winning. Let $b^* = \inf\{b \mid \P[HB(\ce) < b] > 0\}$. Since we assume that players cannot overbid, we know that $b^* \leq v_i < v_1$. 

First, suppose $b^* = v_i$. Then $\P[HB(\ce) < v_i] = 0$. If player 1 bids $(v_i+v_1)/2$ whenever we draw $\bids\sim\ce$ for which $\bids_1 < b^*$ (in case ties are broken in favor of player 1) or whenever $\bids_1 \leq b^*$ (in case ties are broken in favor of player $i$) then player $1$ strictly increases their utility. This contradicts the correlated equilibrium assumption.

Thus we can assume $b^* < v_i$. Define $\tilde b = (b^* + v_i)/2$. Fix a bid $b$ such that $b^* < b < \tilde b$. By assumption we have $\P[HB(\ce) < b] > 0$.  Either player 1 or player $i$ must win by bidding not higher than $b$ with probability at most $\P[HB(\ce) < b]/2$. Let it be player $i$ (otherwise just fill in 1 for $i$ in what follows).  

Consider the following deviating strategy for player $i$: bid $(\tilde b + b)/2$ whenever we draw $\bids\sim\ce$ for which $\bids_i \leq b$ and $\bids_i$ otherwise. For $\bids_i > b$ nothing changes and so the utility stays the same. Next, consider $\bids_i \leq b$.
Player $i$ already won with probability at most $\P[HB(\ce) < b]/2$. 
Now that they bid higher note that the second bid part of the price does not change while the highest bid part goes up by at most $\gamma ((\tilde b + b)/2 - b^*)$ 
On the other hand player $i$ will gain (lower bounding the second highest bid by the highest bid) at least $\P[HB(\ce) < b]/2 \cdot (v_i - (\tilde b + b)/2)$. Net, the utility of player $j$ increases by at least 
\begin{align*}
    \frac{\P[HB(\ce)< b ]}{2} \left(\left(v_i - \frac{\tilde{b} + b}{2}\right) - \left(\frac{\tilde b + b}{2} - b^*\right)\right) &> \\
    \frac{\P[HB(\ce)< b ]}{2} \left(v_i + b^* - 2\tilde{b}\right) &= 0,    
\end{align*}
where we use that $b < \tilde b$. Again, we find a contradiction with the CE conditions. Hence, there cannot be a player $i$ with $v_i < v_1$ having a positive probability of winning implying that the price of anarchy must be 1. \qed
\end{proof}

\myparagraph{Coarse Correlated Price of Anarchy.}
It is known that the coarse correlated price of anarchy for the first-price auction is approximately $1.229$ \cite{feldman2016correlated}, which implies that the result of Theorem~\ref{thm:singleCE} does not extend to coarse correlated equilibria.
We derive the following bound which is good for small values of $\gamma$.

\begin{theorem} \label{thm:CCEsingleItemNpl}
Consider a single-item $\yDPA$ and suppose that bidders cannot overbid. 
Then, the coarse correlated price of anarchy of $\yDPA$ is at most $1/(1-\gamma)$ for all $\gamma \in [0,1)$.
\end{theorem}
\begin{proof}
Let player 1 be the player with highest valuation $v_1$, and if there are multiple players with the highest valuation the player for whom ties are broken in favor when bidding $v_1$.
Let $\Bids$ be a coarse correlated equilibrium. 
We have
\begin{align}
\E_{\bids \sim \Bids}[\SW(\bids)] & = 
\E_{\bids \sim \Bids}[u_1(\bids)] + 
\E_{\bids \sim \Bids}\bigg[\sum_{i \neq 1} u_{i}(\bids) + p^\gamma(\bids) \bigg]. \label{eq:ESW_u1}
\end{align}

Define $\mathcal{E}$ as the event that player $1$ wins the auction with respect to $\Bids$, and let $\bar{\mathcal{E}}$ be the complement event that player $1$ does not win the auction with respect to $\Bids$.

Suppose player $1$ deviates to $v_1$. Then player 1 wins under $(v_{1}, \bids_{-1})$ because either they are the single highest bid or ties are broken in their favor by assumption and no player overbids; note that this holds independently for $\mathcal{E}$ and $\bar{\mathcal{E}}$.
By the $\mcce$ conditions, we thus have 
\begin{align*}
\E_{\bids \sim \Bids}[u_1(\bids)]
 &\ge \E_{\bids \sim \Bids}[u_1(v_1, \bids_{-1})]  \\
 &= (1-\gamma) v_{1} - (1-\gamma) \E_{\bids \sim \Bids} \left[\HB_{-1}(\bids)\right].
\end{align*}
Substituting this inequality in (\ref{eq:ESW_u1}), we obtain 
\begin{align*}
\E_{\bids \sim \Bids}[\SW(\bids)] &\ge 
(1-\gamma) v_{1} - (1-\gamma) \E_{\bids \sim \Bids} \left[\HB_{-1}(\bids)\right] \\
&\quad+ \E_{\bids \sim \Bids}\bigg[\sum_{i \neq 1} u_{i}(\bids) + p^\gamma(\bids) \bigg].
\end{align*}
The proof thus follows if we can show that 
\begin{align}
\E_{\bids \sim \Bids}\bigg[\sum_{i \neq 1} u_{i}(\bids) + p^\gamma(\bids) \bigg]
\ge 
(1-\gamma) \E_{\bids \sim \Bids} \left[\HB_{-1}(\bids)\right].
\label{eq:other-utility-bound}
\end{align}
\underline{Case 1:} Suppose $\bids \in \mathcal{E}$. Then player $1$ wins the auction with respect to $\bids$ and we have 
\[
\sum_{i \neq 1} u_{i}(\bids) + p^\gamma(\bids) = \gamma \bids_{1} + (1-\gamma) \HB_{-1}(\bids) \ge \HB_{-1}(\bids).
\]
\underline{Case 2:} Suppose $\bids \in \bar{\mathcal{E}}$. Then some other player $i' \neq 1$ wins the auction with respect to $\bids$ and we have 
\[
u_{i'}(\bids) + p^\gamma(\bids) = v_{i'} - p^\gamma(\bids) + p^\gamma(\bids) = v_{i'} \ge \bids_{i'} = \HB_{-1}(\bids),
\]
where last inequality holds because $i'$ does not overbid and the last equality holds because $i'$ being the highest bidder implies that $\bids_{i'} = \HB_{-1}(\bids)$.
This concludes the proof. \qed
\end{proof}

Any upper bound for the multi-unit auction setting of course also holds for the single-item setting. 
By combining the bounds of Theorem \ref{thm:OBpoaMUA}, Theorem \ref{thm:CCEpoaMUANOB} and Theorem \ref{thm:CCEsingleItemNpl}, we obtain the upper bound displayed in Figure \ref{fig:bound-overview}(c) for the coarse correlated price of anarchy in the single-item auction setting.

\myparagraph{Coarse Correlated Price of Anarchy for 2-player Auctions.} We now present a more fine-grained picture for the coarse correlated price of anarchy for the 2-player setting. Ultimately, the upper bound for \textit{\mcce-\poa} for two players becomes a combination of three upper bounds, as represented by the three colors in Figure \ref{fig:bound-overview}(d). We already derived the bound we use for small values of $\gamma$ in Theorem \ref{thm:CCEsingleItemNpl}, corresponding to the green graph in the figure. To derive the two remaining bounds, we use an approach inspired by \cite{feldman2016correlated}. The extra difficulty we have is bounding the second-price component. The first-price has a direct relation with winning the auction and so we can use the CCE conditions to bound it while the second-price component is more difficult to get a grip on.
These bounds significantly improve on the bounds of Theorem \ref{thm:OBpoaMUA} and Theorem \ref{thm:CCEpoaMUANOB}. 

First we tackle the interval $\gamma \in [\frac12, 1]$. Note that for $\gamma = 1$ this bound coincides with the (tight) bound in \cite{feldman2016correlated}.

\begin{theorem}\label{thm:CCE2A}
Consider a 2-player single-item $\yDPA$ and suppose that bidders cannot overbid. For $\gamma \in [\frac12, 1]$, the coarse correlated price of anarchy of $\yDPA$ is upper bounded by the blue graph in Figure \ref{fig:bound-overview}(d) (with \textit{\mcce-\poa} $\le 1.295...$ for $\gamma = 0.5$ and \textit{\mcce-\poa} $\le 1.229...$ for $\gamma = 1$).
\end{theorem}

\begin{proof}
Without loss of generality we assume that player 1 has a valuation of $1$ and player 2 has a valuation of $v\leq 1$. Fix $\gamma$ and consider some coarse correlated equilibrium $\Bids$. Let $\alpha = \E[u_1(\Bids)]$ be the utility of player 1 and $\beta = \E[u_2(\Bids)]$ be the utility of player $2$ in $\Bids$. The maximum social welfare is clearly $1$, namely when player 1 wins all the time. Lower bounding the expected welfare of an arbitrary $\Bids$ translates into an upper bound on the price of anarchy. We have
\[
\E[\SW(\Bids)] \geq \alpha + \beta + \E[p^\gamma(\Bids)] = \alpha + \beta + \gamma \E[\HB(\Bids)] + (1-\gamma)\E[\SB(\Bids)]
\]
We try to find the $v, \alpha$ and $\beta$ that minimize this expression and this will then give a lower bound on the expected social welfare. Let $F_X$ be the cumulative distribution function of the random variable $X$ where $X \in \{\HB, \HB_{-1}, \HB_{-2}, \SB\}$. Then by the CCE conditions, and the fact that a CDF is always bounded by 1, we know that
\begin{align}
    F_{\HB_{-1}(\sigma)} (x) &\leq \min\left\{\frac{\alpha}{1-x}, 1\right\} , \ F_{\HB_{-2}(\sigma)} (x) \leq \min\left\{\frac{\beta}{v-x}, 1\right\} \label{eq:HB-i}\\
    F_{\HB(\sigma)}(x) &\leq \min\left\{\frac{\alpha}{1-x}, \frac{\beta}{v-x}, 1\right\} \label{eq:HB}
\end{align}
For example, if $F_{\HB_{-1}(\Bids)} > \frac{\alpha}{1-x}$ and player 1 changes their bid to $x$ their utility will be strictly greater than $\frac{\alpha}{1-x} \cdot (1 - x) = \alpha$ which is more than their current utility contradicting the CCE conditions. 

Also note that $\alpha \geq 1-v$ because player 1 bidding $v+\epsilon$ will yield a utility of at least $1-v-\epsilon$ for any positive $\epsilon$. The other player is not allowed to bid above $v$, thus player 1 always wins when bidding $v+\epsilon$.

Observe that for $n = 2$ players the following chain of equalities holds
\begin{align}
    F_{\SB(\Bids)}(x) &= \P[\SB(\Bids) \leq x] \nonumber\\
    &= \P[\min(\HB_{-1}(\Bids), \HB_{-2}(\Bids)) \leq x] \nonumber\\
    &= \P[\HB_{-1}(\Bids)\leq x] + \P[\HB_{-2](\Bids)}\leq x] - \P[HB(\Bids)\leq x] \nonumber\\
    &= F_{\HB_{-1}(\Bids)}(x) + F_{\HB_{-2}(\Bids)}(x)  - F_{\HB(\Bids)}(x). \label{FSBUB}
\end{align}

Let us get a more explicit expression for the expected payment using (\ref{FSBUB})
\begin{align}
    \E[p^\gamma(\Bids)]&=\gamma \E[\HB(\Bids)] + (1-\gamma)\E[\SB(\Bids)] \nonumber\\
    &= \gamma\int_0^1 1 - F_{\HB(\Bids)}(x) dx +  (1-\gamma)\int_0^1 1 - F_{\SB(\Bids)}(x) dx  \label{Payment}\\
    &= \gamma\int_0^1 1 - F_{\HB(\Bids)}(x) dx +  (1-\gamma)\cdot \nonumber \\
    & \quad\int_0^1 1 - F_{\HB_{-1}(\Bids)}(x) -F_{\HB_{-2}(\Bids)}(x) + F_{\HB(\Bids)}(x) dx  \nonumber\\
    &= (2\gamma - 1) \int_0^1 1-  F_{\HB(\Bids)}(x) dx + \nonumber \\
    &\quad (1-\gamma)\sum_{i=1}^2 \int_0^1 1 - F_{\HB_{-i}(\Bids)}(x)dx \nonumber
\end{align}

Using the two bounds in (\ref{eq:HB-i}) we can lower bound the two integrals in the summation
\begin{align*}
    \int_0^1 1 - F_{\HB_{-1}(\Bids)}(x) dx \geq \int_0^{1-\alpha} 1-\frac{\alpha}{1-x} dx = 1 - \alpha + \alpha\ln(\alpha) \\
    \int_0^1 1 - F_{\HB_{-2}(\Bids)}(x) dx \geq \int_0^{v-\beta} 1 -\frac{\beta}{v-x} dx = v - \beta + \beta\ln(\beta/v) 
\end{align*}
If $\gamma \geq \frac12$ then $2\gamma - 1 \geq 0$ and so we can use (\ref{eq:HB}) to lower bound the integral on the left by
\begin{align*}
    \int_0^1 1 - F_{\HB(\Bids)}(x) dx \geq \int_0^1 1 - \min\left\{\frac{\alpha}{1-x}, \frac{\beta}{v-x}, 1\right\} dx
\end{align*}
We split up in two cases. \\ \\
\underline{Case 1:} $\beta \geq v\alpha$. Then $\frac{\beta}{v-x} \geq \frac{\alpha}{1-x}$ for all $x \in [0,v]$ and so
\begin{equation}
\int_0^1 1 - \min\left\{\frac{\alpha}{1-x}, \frac{\beta}{v-x}, 1\right\} dx = \int_0^{1-\alpha}1 - \frac{\alpha}{1-x} dx = 1 - \alpha + \alpha\ln(\alpha)
\end{equation}
giving a lower bound on expected welfare of
\begin{align*}
\E[\SW(\Bids)] \geq \alpha &+ \beta + (2\gamma - 1)(1 - \alpha + \alpha \ln(\alpha))\\ &+(1-\gamma)\left(1 - \alpha+\alpha\ln(\alpha) + v - \beta+ \beta\ln(\beta/v)\right) 
\end{align*}
Using that $v \geq 1-\alpha$ and $\beta \geq v\alpha \geq \alpha(1-\alpha)$ this is lower bounded by
\begin{align} \label{LB1}
    1 + \alpha(1-\alpha)\ln(\alpha) + \gamma \alpha(1 - \alpha + \alpha\ln(\alpha))
\end{align}
\underline{Case 2:} $\beta < v\alpha$. First $\frac{\beta}{v-x}$ is smaller than $\frac{\alpha}{1-x}$ until $x = \theta = \frac{\alpha v - \beta}{\alpha-\beta}$ when $\frac{\beta}{v-x}$ takes over. In this case the integral is bounded from below by
\begin{align*}
\int_0^1 1 - \min\left\{\frac{\alpha}{1-x}, \frac{\beta}{v-x}, 1\right\} dx &\geq \\
\int_0^\theta 1 - \frac{\beta}{v-x} dx + \int_\theta^{1-\alpha} 1 - \frac{\alpha}{1-x} dx &=\\
\alpha \ln\left(\frac{\alpha-\beta}{1-v}\right) + 1 - \alpha + \beta\ln\left(\frac{\beta(1-v)}{v(\alpha-\beta)}\right) &
\end{align*}
which gives a lower bound on the expected social welfare of
\begin{align}
\label{LB2}
\E[\SW(\Bids)] & \geq \alpha +\beta \\
&\quad+(2\gamma-1)\left(\alpha \ln\left(\frac{\alpha-\beta}{1-v}\right) + 1 - \alpha + \beta\ln\left(\frac{\beta(1-v)}{v(\alpha-\beta)}\right)\right)  \notag \\
&\quad + (1-\gamma)\left(1 -\alpha +\alpha\ln(\alpha) + v - \beta + \beta\ln(\beta/v)\right) \notag
\end{align}
The derivative with respect to $v$ is
\[
(2\gamma -1)\frac{\alpha v- \beta}{(1-v)v} + (1-\gamma)(1-\beta/v)
\]
For $\beta< v\alpha$ this is positive and thus the minimum is attained when $v$ is smallest, i.e., $v=1-\alpha$. Substituting that in (\ref{LB2}) gives

\begin{align}
\E[\SW(\Bids)] &\geq \alpha +\beta \nonumber \\
&\quad+(2\gamma-1)\left(\alpha \ln\left(\frac{\alpha-\beta}{\alpha}\right) + 1 - \alpha + \beta\ln\left(\frac{\beta\alpha}{(1-\alpha)(\alpha-\beta)}\right)\right)  \nonumber \\
&\quad + (1-\gamma)\left(1 -\alpha +\alpha\ln(\alpha) + 1-\alpha - \beta + \beta\ln\left(\frac{\beta}{1-\alpha}\right)\right) \nonumber  \\
&= 1 + \gamma \beta + (2\gamma-1)(\alpha-\beta) \ln(\alpha-\beta) \nonumber \\
&\quad+ (2-3\gamma)\alpha\ln(\alpha) + (2\gamma-1)\beta\ln(\alpha)  \label{eq:ESW2}\\
& \quad + \gamma \beta\ln(\beta) - \gamma \beta \ln(1-\alpha) \nonumber
\end{align}
Note that filling in $\beta = \alpha(1-\alpha)$ yields the same revenue as in (\ref{LB1}). Changing $\beta < v\alpha$ to $\beta \leq v\alpha = \alpha(1-\alpha)$ subsumes case 1. So we only have to find the minimum in case 2.

The derivative of (\ref{eq:ESW2}) with respect to $\beta$ is
\begin{align*}
(2\gamma-1)\ln\left(\frac{\alpha}{\alpha-\beta}\right) + \gamma \ln\left(\frac{\beta}{1-\alpha}\right) + 1
\end{align*}
This becomes $0$ when
\begin{align*}
\ln\left(\frac{\alpha^{2\gamma-1}\beta^\gamma}{(\alpha-\beta)^{2\gamma-1}(1-\alpha)^\gamma}\right) = -1 
&\iff \frac{\beta^\gamma}{(\alpha-\beta)^{2\gamma-1}} - \frac{(1-\alpha)^\gamma}{e \alpha^{2\gamma-1}} = 0
\end{align*}
For fixed $\alpha$ the expression on the left is negative for $\beta$ close to 0, and positive for $\beta$ close to $\alpha$. Also the second derivative with respect to $\beta$ is always positive on $[0,\beta]$. Thus we can use binary search to quickly find $\beta$ satisfying the equality. Call this $\beta_\alpha$.\footnote{$\beta_\alpha$ may violate the case assumption that $\beta \leq v\alpha$ but removing this restriction can only decrease the minimum value of the expected social welfare.}

Then we have
\begin{align*}
    \E[\SW(\Bids)] &\geq  1 + \gamma \beta_\alpha + (2\gamma-1)(\alpha-\beta_\alpha) \ln(\alpha-\beta_\alpha) \\
    &\quad+ (2-3\gamma)\alpha\ln(\alpha) + (2\gamma-1)\beta_\alpha\ln(\alpha) \\
& \quad + \gamma \beta_\alpha\ln(\beta_\alpha) - \gamma \beta_\alpha \ln(1-\alpha)
\end{align*}

For $\gamma = 1/2$ we compute $\beta_\alpha = (1-\alpha)/{e^2}$ and then the social welfare is minimized for $\alpha = e^{-1-1/e^2} \approx 0.3213...$ with value $0.7716...$. While for $\gamma = 1$ we have $\beta_\alpha = {\alpha(1-\alpha)}/(1-\alpha+e\alpha)$ where the social welfare is minimized for $\alpha \approx 0.2743...$ with value $0.8135...$. In both cases (\ref{eq:ESW2}) becomes a unimodal function. Plotting (\ref{eq:ESW2}) for various values of $\alpha$, when doing binary search to find $\beta_\alpha$ as a subroutine, suggests that this is the case for all $\gamma$. Making this assumption we can use a ternary search on $\alpha$ with a binary search to find $\beta_\alpha$ as a subroutine to quickly find the minimum.
Finally, taking 1 over this value gives us an upper bound on the price of anarchy, presented as the blue graph in Figure \ref{fig:bound-overview}(d). \qed
\end{proof}

The previous theorem holds for $\gamma \in [\frac12, 1]$. With a similar proof template, making use of an upper bound on the highest bid, we can derive an upper bound on the coarse correlated price of anarchy for the lower to mid range of $\gamma$
\begin{theorem}\label{thm:CCE2B}
Consider a 2-player single-item $\yDPA$ and suppose that bidders cannot overbid. For $\gamma \in (0.217..., \frac12]$, the coarse correlated price of anarchy of $\yDPA$ is upper bounded by the orange graph in Figure \ref{fig:bound-overview}(d) (with \textit{\mcce-\poa} $\le 1.515...$ for intersection point $\gamma = 0.339...$ and \textit{\mcce-\poa} $\le 1.295...$ for $\gamma = 0.5$).
\end{theorem}
\begin{proof}
For $\gamma \in [0, 1/2]$, note that in the final equality of (\ref{Payment}), we have $(2\gamma - 1) \le 0 $. To \textit{lower} bound the social welfare, we should therefore \textit{upper} bound the expected highest bid. For this, note that due to the fact that players cannot overbid, player 2 never bids higher than $v$. Therefore, for any $\gamma > 0$, any bid of player 1 that is (strictly) above $v$ is (strictly) dominated by bidding $v$ instead\footnote{Formally, player 1 should bid $v+\epsilon$ for any $\epsilon>0$. Since $\epsilon$ can be an arbitrarily small number, we ignore it in the remainder of the proof for notational  convenience.}. Using this, it is clear that $ \E[\HB(\Bids)] \le v$. Again using (\ref{eq:HB-i}) to lower bound the two rightmost integrals, we get 
\begin{align*}
    &\gamma \E[\HB(\Bids)] + (1-\gamma)\E[\SB(\Bids)] \\
    &\geq (2\gamma - 1) \int_0^1 1-  F_{\HB(\Bids)}(x) dx \\
    &\quad+ (1-\gamma)\left(\int_0^1 1 - F_{\HB_{-1}(\Bids)}(x)dx + \int_0^1 1 - F_{\HB_{-2}(\Bids)}(x)dx\right) \\
    &\geq (2\gamma - 1) v + (1-\gamma) (1 - \alpha + \alpha\ln(\alpha) + v - \beta + \beta\ln(\beta/v)),
\end{align*}
so that 
\begin{align} \label{smallergammaUB}
    \E[\SW(\Bids)] &\geq \alpha + \beta + (2\gamma - 1) v \\
    &\quad + (1-\gamma) (1 - \alpha + \alpha\ln(\alpha) + v - \beta + \beta\ln(\beta/v)) \nonumber \\
    &= \gamma (\alpha + \beta + v) + (1-\gamma) (1 + \alpha\ln(\alpha) + \beta\ln(\beta) - \beta\ln(v)). \nonumber
\end{align}
The derivative of this bound with respect to $\beta$ equals
\[
    \gamma + (1-\gamma) (1+ \ln(\beta) - \ln(v)) = 1+ (1-\gamma) \ln(\beta /v).
\]
Note that this derivative is equal to zero for $\beta = v e^{-1/(1-\gamma)}$, and that it is positive for greater $\beta$ and negative for smaller $\beta$. Therefore, the bound attains its minimum at $\beta = v e^{-1/(1-\gamma)}$. Substituting this $\beta$ in (\ref{smallergammaUB}) yields 
\begin{align} \label{eq:betaproofsmallgammaUB}
    \E[\SW(\Bids)] &\geq \gamma (\alpha + (1+e^{-\frac{1}{1-\gamma}})v) \\
    \quad &+ (1-\gamma) (1 + \alpha\ln(\alpha) + v e^{-\frac{1}{1-\gamma}} \ln( e^{-\frac{1}{1-\gamma}} ) \nonumber \\
    &= \gamma (\alpha + (1+ e^{-\frac{1}{1-\gamma}}) v) + (1-\gamma)(1 + \alpha\ln(\alpha)) - v e^{-\frac{1}{1-\gamma}}.\nonumber
\end{align}
Next, we take the derivative of (\ref{eq:betaproofsmallgammaUB}) with respect to $v$, which gives 
\[
    \gamma (1+e^{-\frac{1}{1-\gamma}}) - e^{-\frac{1}{1-\gamma}} = \gamma - (1-\gamma) e^{-\frac{1}{1-\gamma}}.
\]
This derivative is positive for all $\gamma > 0.21781\dots$, so for all $\gamma \in (0.21781\dots, 1/2]$, we minimize the upper bound by setting $v$ to its lowest admissible value, being $v = 1 -\alpha$ (due to the CCE condition that $\alpha \ge 1 -v$). \\ \\
Substituting the optimal parameter settings $\beta = v e^{-1/(1-\gamma)} = (1-\alpha) e^{-1/(1-\gamma)}$ gives the following social welfare bound 
\begin{align} \label{eq:alphaproofsmallgammaUB}
    \E[\SW(\Bids)] &\geq \gamma (\alpha + (1-\alpha) (1+ e^{-\frac{1}{1-\gamma}})) + (1-\gamma)(1 + \alpha\ln(\alpha)) - \nonumber \\
    &\quad (1-\alpha) e^{-\frac{1}{1-\gamma}} \nonumber \\
    &= 1 - (1-\gamma) (1-\alpha) e^{-\frac{1}{1-\gamma}} + (1-\gamma)\alpha\ln(\alpha),
\end{align}
as a function of $\alpha$ only, which we optimize by setting its derivative with respect to $\alpha$ equal to zero. This yields
\begin{align*}
(1-\gamma) e^{-\frac{1}{1-\gamma}} + (1-\gamma)(1 + \ln(\alpha)) = 0 &\iff  \ln(\alpha) = -1 -e^{-\frac{1}{1-\gamma}} \\
&\iff \alpha = e^{-1 -e^{-1/(1-\gamma)} }.    
\end{align*}
To facilitate the simplification of the formula of the final bound, we first substitute only $\ln(\alpha)$ in (\ref{eq:alphaproofsmallgammaUB}), after which $\alpha$ itself is substituted in the final step. We get 
\begin{align} \label{eq:finalpfsmallgamma}
    \E[\SW(\Bids)] &\geq 1 - (1-\gamma)e^{-\frac{1}{1-\gamma}} + (1-\gamma) \alpha  e^{-\frac{1}{1-\gamma}} + \nonumber \\
    &\quad (1-\gamma)\alpha (-1 -e^{-\frac{1}{1-\gamma}}) \nonumber \\
    &= 1 - (1-\gamma)e^{-\frac{1}{1-\gamma}} - (1-\gamma) \alpha \nonumber \\
    &= 1- (1-\gamma) (e^{-\frac{1}{1-\gamma}} +  e^{-1 -e^{-1/(1-\gamma)} }).
\end{align}
We divide 1 by (\ref{eq:finalpfsmallgamma}) to get the upper bound on the price of anarchy presented as the orange graph in Figure \ref{fig:bound-overview}(d). \qed
\end{proof}

\begin{toappendix}
\myparagraph{Extensions to $n$ players:} 
We can modify the procedure in Theorems \ref{thm:CCE2A} and \ref{thm:CCE2B} to work for $n$ players, producing a different graph for every $n$. Observe that a more general version of (\ref{FSBUB}) holds:
\begin{align*}
F_{SB(\Bids)}(x) &= \sum_{i=1}^n F_{HB_{-i}(\Bids)}(x) - (n-1) F_{HB(\Bids)}(x) \\
&\leq \sum_{i=1}^2 F_{HB_{-i}(\Bids)}(x) - (n-1) F_{HB(\Bids)}(x) \\
\end{align*}
We can use this to bound the second highest bid. Doing a similar analysis as in Theorem \ref{thm:CCE2A} gives us a procedure that works for $\gamma \in \left[\frac{n-1}{n}, 1\right]$ and similarly Theorem \ref{thm:CCE2B} yields a procedure for $\gamma\in \left(0.217..., \frac{n-1}{n}\right]$. We obtain bounds which are always tight at $\gamma=1$ for any $n$, but the higher $n$ the faster the bound goes up when $\gamma$ gets away from $1$, and already for $n=4$ the multi-unit bound lies below what the extension of Theorem \ref{thm:CCE2B} can give us. 
\end{toappendix}

\section{Conclusion and Future Work}

Our bound on the \mcce-\poa\ of \yFPA\ is tight over the entire range of $\gamma \in [0, 1]$ if players can overbid, both in the single-item and multi-unit auction setting. Despite the fact that our bounds on the \mcce-\poa\ are rather low already if players cannot overbid, further improvements might still be possible. We consider this a challenging open problem for future work.

On a more conceptual level, in this paper we considered a basic bid rigging model where the auctioneer colludes with the winning bidders only. It will be very interesting to study the price of anarchy of more complex bid rigging models; for example, the model introduced in \cite{LW2010} (ideally generalized to the multi-unit auction setting) might be a natural next step.

\bibliographystyle{splncs04} 
\bibliography{references}

\renewcommand{\appendixbibliographystyle}{splncs04}

\end{document}